\documentclass[10pt,a4paper]{article}
\usepackage{amssymb,amsmath}
\usepackage{amsthm}
\usepackage{bbm}
\usepackage{stmaryrd}
\usepackage{hyperref}
\usepackage[balance,small,nohug]{diagrams}
\usepackage[usenames,dvipsnames]{xcolor}

\input xy
\xyoption{all} 

\numberwithin{equation}{section}

\newtheorem{theorem}{Theorem}[subsection]

\newtheorem{proposition}[theorem]{Proposition}

\theoremstyle{definition}
\newtheorem{definition}[theorem]{Definition}
\newtheorem{remark}[theorem]{Remark}
\newtheorem{example}[theorem]{Example}

\newcommand*{\longhookrightarrow}{\ensuremath{\lhook\joinrel\relbar\joinrel\rightarrow}}

\DeclareMathOperator{\Vect}{Vect}

\DeclareMathOperator{\w}{w}

\DeclareMathOperator{\ad}{ad}

\newarrow{Relation}{-}{-}{-}{-}{triangle}

\font\black=cmbx10 \font\sblack=cmbx7 \font\ssblack=cmbx5 \font\blackital=cmmib10  \skewchar\blackital='177
\font\sblackital=cmmib7 \skewchar\sblackital='177 \font\ssblackital=cmmib5 \skewchar\ssblackital='177
\font\sanss=cmss10 \font\ssanss=cmss8 
\font\sssanss=cmss8 scaled 600 \font\blackboard=msbm10 \font\sblackboard=msbm7 \font\ssblackboard=msbm5
\font\caligr=eusm10 \font\scaligr=eusm7 \font\sscaligr=eusm5  \font\fraktur=eufm10
\font\sfraktur=eufm7 \font\ssfraktur=eufm5 
\font\bsymb=cmsy10 scaled\magstep2
\def\all#1{\setbox0=\hbox{\lower1.5pt\hbox{\bsymb
       \char"38}}\setbox1=\hbox{$_{#1}$} \box0\lower2pt\box1\;}
\def\exi#1{\setbox0=\hbox{\lower1.5pt\hbox{\bsymb \char"39}}
       \setbox1=\hbox{$_{#1}$} \box0\lower2pt\box1\;}

\def\tx#1{{\fam0\relax#1}}

\newfam\bifam
\textfont\bifam=\blackital \scriptfont\bifam=\sblackital \scriptscriptfont\bifam=\ssblackital

\newfam\blfam
\textfont\blfam=\black \scriptfont\blfam=\sblack \scriptscriptfont\blfam=\ssblack

\newfam\bbfam
\textfont\bbfam=\blackboard \scriptfont\bbfam=\sblackboard \scriptscriptfont\bbfam=\ssblackboard

\newfam\ssfam
\textfont\ssfam=\sanss \scriptfont\ssfam=\ssanss \scriptscriptfont\ssfam=\sssanss
\def\sss#1{{\fam\ssfam\relax#1}}

\newfam\clfam
\textfont\clfam=\caligr \scriptfont\clfam=\scaligr \scriptscriptfont\clfam=\sscaligr

\newfam\frfam
\textfont\frfam=\fraktur \scriptfont\frfam=\sfraktur \scriptscriptfont\frfam=\ssfraktur

\def\hpb#1{\setbox0=\hbox{${#1}$}
    \copy0 \kern-\wd0 \kern.2pt \box0}
\def\vpb#1{\setbox0=\hbox{${#1}$}
    \copy0 \kern-\wd0 \raise.08pt \box0}

\def\pmb#1{\setbox0\hbox{${#1}$} \copy0 \kern-\wd0 \kern.2pt \box0}
\def\pmbb#1{\setbox0\hbox{${#1}$} \copy0 \kern-\wd0
      \kern.2pt \copy0 \kern-\wd0 \kern.2pt \box0}
\def\pmbbb#1{\setbox0\hbox{${#1}$} \copy0 \kern-\wd0
      \kern.2pt \copy0 \kern-\wd0 \kern.2pt
    \copy0 \kern-\wd0 \kern.2pt \box0}
\def\pmxb#1{\setbox0\hbox{${#1}$} \copy0 \kern-\wd0
      \kern.2pt \copy0 \kern-\wd0 \kern.2pt
      \copy0 \kern-\wd0 \kern.2pt \copy0 \kern-\wd0 \kern.2pt \box0}
\def\pmxbb#1{\setbox0\hbox{${#1}$} \copy0 \kern-\wd0 \kern.2pt
      \copy0 \kern-\wd0 \kern.2pt
      \copy0 \kern-\wd0 \kern.2pt \copy0 \kern-\wd0 \kern.2pt
      \copy0 \kern-\wd0 \kern.2pt \box0}

\mathchardef\za="710B  
\mathchardef\zb="710C  
\mathchardef\zg="710D  
\mathchardef\zd="710E  
\mathchardef\zve="710F 
\mathchardef\zz="7110  
\mathchardef\zh="7111  
\mathchardef\zvy="7112 
\mathchardef\zi="7113  
\mathchardef\zk="7114  
\mathchardef\zl="7115  
\mathchardef\zm="7116  
\mathchardef\zn="7117  
\mathchardef\zx="7118  
\mathchardef\zp="7119  
\mathchardef\zr="711A  
\mathchardef\zs="711B  
\mathchardef\zt="711C  
\mathchardef\zu="711D  
\mathchardef\zvf="711E 
\mathchardef\zq="711F  
\mathchardef\zc="7120  
\mathchardef\zw="7121  
\mathchardef\ze="7122  
\mathchardef\zy="7123  
\mathchardef\zf="7124  
\mathchardef\zvr="7125 
\mathchardef\zvs="7126 
\mathchardef\zf="7127  
\mathchardef\zG="7000  
\mathchardef\zD="7001  
\mathchardef\zY="7002  
\mathchardef\zL="7003  
\mathchardef\zX="7004  
\mathchardef\zP="7005  
\mathchardef\zS="7006  
\mathchardef\zU="7007  
\mathchardef\zF="7008  
\mathchardef\zW="700A  
\mathchardef\zC="7009  

\newcommand{\be}{\begin{equation}}
\newcommand{\ee}{\end{equation}}
\newcommand{\ra}{\rightarrow}

\newcommand{\bea}{\begin{eqnarray}}
\newcommand{\eea}{\end{eqnarray}}
\def\*{{\textstyle *}}
\newcommand{\R}{{\mathbb R}}

\newcommand{\so}{\textstyle{so}(3,\R)}

\newcommand{\D}{{\rm d}}

\newcommand{\we}{\wedge}
\newcommand{\nn}{\nonumber}

\newcommand{\s}{{\textstyle *}}

\newcommand{\pa}{\partial}
\newcommand{\ti}{\times}
\newcommand{\A}{{\cal A}}

\def\ran{\rangle}

\def\Aff{\sss{Aff}}

\def\cD{{\mathcal{D}}}
\def\cR{{\mathcal{R}}}
\def\bA{\mathbf{A}}

\def\wh{\widehat}
\def\wt{\widetilde}

\def\Sec{\sss{Sec}}

\def\AV{\sss{AV}}

\def\Z{\mathbf{Z}}

\def\la{\langle}
\def\ran{\rangle}
\def\Mi{\textnormal{Mi}}

\def\sP{{\sss P}}
\def\sZ{{\sss Z}}

\def\sT{{\sss T}}

\def\sV{{\sss V}}

\def\st{{\sss t}}

\def\sv{{\sss v}}

\def\xi{\tx{i}}

\def\cM{\cal M}

\def\cD{\cal D}

\def\xd{\operatorname{d}}

\def\s*{{\scriptstyle *}}

\def\cP{\mathcal{P}}

\def\cM{\mathcal{M}}

\def\uxd{{\underline{\mathrm{d}}}}

\newcommand{\beas}{\begin{eqnarray*}}
\newcommand{\eeas}{\end{eqnarray*}}
\newdir{|>}{%
!/4.5pt/@{|}*:(1,-.2)@^{>}*:(1,+.2)@_{>}}
\newdir{ (}{{}*!/-5pt/@^{(}}

\def\g{\mathfrak{g}}
\def\half{\frac{1}{2}}

\begin{document}
\bibliographystyle{plain}

\author{\\
        Andrew James Bruce$^1$\\ Katarzyna  Grabowska$^2$\\ Janusz Grabowski$^1$\\
        \\
         $^1$ {\it Institute of Mathematics}\\
                {\it Polish Academy of Sciences} \\
         $^2$ {\it Faculty of Physics}\\
                {\it University of Warsaw}
                }

\date{\today}
\title{Higher order   mechanics on graded bundles
\thanks{Research of K.~Grabowska and J.~Grabowski founded by the  Polish National Science Centre grant under the contract number DEC-2012/06/A/ST1/00256. A.J.~Bruce graciously acknowledges the financial support of the Warsaw Centre for Mathematics and Computer Science in the form of a postdoctoral fellowship. }}
\maketitle

\begin{abstract}

In this paper we develop a geometric approach to higher order mechanics on graded bundles in both, the Lagrangian and Hamiltonian formalism, via  the recently discovered weighted algebroids. We present the corresponding Tulczyjew triple for this higher order situation and derive in this framework the phase equations from an arbitrary (also singular) Lagrangian or Hamiltonian, as well as the Euler--Lagrange equations.  As important examples, we  geometrically derive the classical higher order Euler--Lagrange equations and analogous reduced equations for invariant higher order Lagrangians on Lie groupoids.
\end{abstract}

\begin{small}
\smallskip

\end{small}

\tableofcontents

\section{Introduction}\label{sec:Intro}

It is certainly true that higher derivative theories have received less mathematical attention than first order theories and for some sound reasons. Recall the famous Ostrogradski Theorem;\medskip

\noindent \emph{If a higher order time derivative Lagrangian is non-degenerate,  then there is at least one linear instability in the Hamiltonian of this system.}
\medskip

In this context non-degeneracy means that the highest derivative term can be expressed in terms of canonical variables. The Ostrogradski instability leads to the fact that the associated Hamiltonian is not bounded from below. This by itself is not a problem as classically one can only measure energy differences. The difficulty lies in including quantum effects of continuum theories with interactions. In such quantum field theories an ``empty'' state would spontaneously decay into a collection of positive and negative energy particles, in accordance with energy conservation. Moreover, there are typically states with non-positive norm known as ghosts that require exorcising from the physical theory. Other problems include the presence of extra degrees of freedom and the well-posedness question of an initial-value formulation  of the equations of motion. Such issues are of course not a really problem if we consider the theory to be an \emph{effective theory}; at some scale ``new physics'' enters the picture and this not dependent on higher order derivatives. Or in other words the effective theory is a kind of truncation of a complete theory that is first order.  Indeed, the main source of higher order Lagrangians in particle physics and cosmology is via effective theories. In fact for effective theories one can use the equations of motion to reduce the order and reduce the theory to first order. For mechanical systems, higher order theories similarly arise as phenomenological models via various assumptions made in the modelling.  In the context of classical mechanics,  higher derivative theories will exhibit pathological derivations from Newton's laws, but as we do not have external interactions here with continuous degrees of freedom they are not fundamentally sick as higher derivative quantum field theories. In short, higher derivative theories are important in physics even if they cannot be considered as fundamental theories. For a very nice review of the Ostrogradski instability the reader can consult \cite{Woodard:2007}.\smallskip

One important technical distinction here is that higher derivative theories with \emph{degenerate Lagrangians} are  healthy in the sense that they do not suffer with the Ostrogradski instability. Degeneracy is an interesting feature of physical models as all theories that possess continuous local symmetries are degenerate, independently of the order. A closely related subject is the reduction of Lagrangian systems invariant under the action of some group(oid).  \smallskip

 The challenge of describing mechanical systems on Lie groupoids and their reduction to Lie algebroids   was first posted by Weinstein \cite{Weinstein:1996}. Many authors took up this challenge, for example see \cite{Carinena:2001,Cortes:2006,Cortes:2009,Grabowska:2006,Grabowska:2008,Martinez:2001}. Extending the geometric tools of the Lagrangian formalism on tangent bundles to Lie algebroids was motivated by the fact that reductions usually push one out of the environment of tangent bundles. In a similar way, reductions of higher order tangent bundles, which is where higher order mechanical Lagrangians ``live'',  will push one into the environment of ``higher Lie  algebroids".\smallskip

Thus, general geometric methods that can handle both non-degenerate and degenerate  higher order Lagrangians  are very desirable in mathematical physics. In this paper we show how one can use weighted  skew/Lie algebroids as described in \cite{Bruce:2014} to achieve this goal in the context of the Lagrangian formulation of classical mechanics.  Furthermore, we complete the Tulczyjew triple and present the complimentary Hamiltonian formalism. The mechanical systems we will be dealing with will be rather general, we will make no assumption about the degeneracy for example, and the generalised higher order velocities will take their values in \emph{graded bundles}. Graded bundles are a natural higher order generalisation of the notion of a vector bundle, and were first defined and  studied by Grabowski \& Rotkiewicz in \cite{Grabowski2012}. A canonical example of a graded bundle is the k-th order tangent bundle and so our constructions can handle standard higher order Lagrangians. A geometric understanding of Lagrangian mechanics on higher tangent bundles  close to our way of thinking was developed by de Leon \& Lacomba \cite{Leon:1989} and independently by Crampin \cite{Crampin:1990}.\smallskip

A little more exotically, we show that higher order Lagrangian mechanics on Lie algebroids  can naturally be accommodated within our framework. The study of such systems  is motivated by the study of higher order systems on Lie groupoids with Lagrangians invariant under the groupoid multiplication. However, Lagrangian systems on Lie algebroids need not arise as the reduction of Lagrangian systems on Lie groupoids, and indeed not all Lie algebroids integrate to a Lie groupoid.   Interestingly, higher derivative mechanics on Lie groupoids has received little attention in the literature. In fact we are only aware of  two works in this direction; \cite{Colombo:2013,Jozwikowski:2014} who both derive the second order Euler--Lagrange equations on a Lie algebroid. The Lie group case has similarly not received a lot of attention;  the so called higher Euler--Poincar\'{e} equations were derived in \cite{Gay-Balmaz:2012}.\smallskip

We do not consider the development of higher order mechanics on Lie groupoids and algebroids as a purely academic exercise. Recall that many interesting mechanical systems can be understood as the Euler--Poincar\'{e} equations on a Lie algebra.  Similarly, a Lagrangian  on a principal $G$-bundle that is invariant under the action of $G$ leads to Lagrangian system on the associated Atyiah bundle, which is canonically a Lie algebroid.  In this paper we will provide a rather general setup that allows higher order versions of the Euler--Poincar\'{e} and Lagrange--Poincar\'{e}  equations to appear geometrically.  We certainly envisage applications of this geometric  setup, say via optimal control theory (with symmetries), in  the fields of engineering and the physical sciences.\smallskip

Our approach to higher order mechanics  using weighted algebroids makes use of first order mechanics on Lie algebroids subject to affine (vakonomic) constraints. The higher order flavour arises as underlying a weighted algebroid is a graded bundle.  This mimics the  approach to higher order mechanics on $\sT^{k}M$ by studying  first order mechanics on $\sT(\sT^{k-1}M)$ and then using the natural embedding as a constraint.\smallskip

As our approach is geometric we will understand the phase dynamics as \emph{implicit dynamics}; this is quite standard  for vakonomic constraints. That is  the phase dynamics is a subset of the tangent bundle of the phase space of the system in question. Quite  often people are satisfied with  just the Euler--Lagrange equations rather than the full phase dynamics.  However,  the Euler--Lagrange equations, being only a `shadow' of phase dynamics (Lagrange equations), do not carry the important information on how momenta are associated with velocities. Also, their solutions can come by gluing different phase solutions, therefore, apart from `good' cases, they are physically unsatisfactory.
\smallskip

We remark that various  higher order versions of Lie algebroids besides weighted algebroids have  appeared in the literature. First we must mention the (prototype) higher Lie algebroids as defined by J\'o\'zwikowski \& Rotkiewicz \cite{Jozwikowski:2014}, which are a direct generalisation of the \emph{kappa-relation} $\kappa : \sT E \rRelation \sT E$ for Lie algebroids. Although this approach is motivated by geometric mechanics, it seems not to be quite suitable for the approach pursued in this paper.\smallskip

Secondly, there is the more established notion of a higher or nonlinear Lie algebroid as defined by Voronov \cite{Voronov:2001qf,voronov-2010} in terms of a weight one homological vector field on a non-negatively graded supermanifold. As it stands, it is not clear that this notion is related to the reduction of higher tangent bundles on Lie groupoids nor how it can be applied in geometric mechanics. As we can view weighted algebroids as a special class of Voronov's nonlinear algebroids (c.f. \cite{Bruce:2014}) the link with higher order tangent bundles is clear. Thus in part, this paper establishes a link between geometric mechanics and nonlinear algebroids via weighted algebroids. However, we will not make any use of supermanifolds in the description of weighted algebroids.\smallskip

It seems that  higher order mechanical systems on manifolds with Lie algebroid-like structures are a potentially rich subject awaiting to be harvested.

\medskip

\noindent \textbf{Arrangement of paper:} In section \ref{sec:GradedBundles} we briefly recall the necessary parts of theory of graded bundles, including the notion of the linearisation functor and the linear dual, as needed in this work. We also present the bare minimum of the theory of weighted algebroids as tailored to the needs of the current paper. For a more complete account the reader is urged to consult \cite{Bruce:2014}. In essence this first section if a review of the graded geometry needed in this paper. In section \ref{sec:PhaseSpacesMechanicsConstraints} we review some nonstandard aspects of geometric mechanics including affine phase spaces, mechanics on Lie algebroids and vakonomic constraints. Our review here is focused on the elegant approach to geometric mechanics as pioneered by Tulczyjew. With this in mind we end that section we the Tulczyjew triple for standard higher order mechanics. The new elements of this paper are to be found in section \ref{sec:MechanicsGradedBundles}, where we present the Lagrangian and Hamiltonian formalism on graded bundles via weighed algebroids, as well as the corresponding  Tulczyjew triple.  We then look at an application of our formalism to higher order mechanics on a Lie algebroid in section \ref{sec:Higher Lagrangian Algebroid}. In particular we explicitly construct the higher order Euler--Lagrange equations for a higher order Lagrangian system on a Lie algebroid. As particular examples, we geometrically derive the second order Euler--Poincar\'{e} and Lagrange--Poincar\'{e}  equations.


\section{Graded bundles and weighted algebroids}\label{sec:GradedBundles}
In this section we briefly recall parts of the theory of graded bundles, n-tuple graded bundles, the linearisation functor and weighted algebroids as needed in later sections of this paper. The interested reader should consult the original literature \cite{Bruce:2014,Grabowski:2006,Grabowski2012} for details, such as proofs of the statements made in this section.

\subsection{Graded bundles}\label{subsec:GradedBundles}
\noindent \textbf{Graded an n-tuple graded bundles:} Manifolds and supermanifolds that carry various extra gradings on their structure sheaf are now an established part of modern geometry and mathematical physics. The general theory of graded manifolds in our understanding was initiated by Voronov  in $\cite{Voronov:2001qf}$. The graded structure on such (super)manifolds is conveniently encoded in a weight vector field whose action via the Lie derivative counts the degree of tensor and tensor-like objects on the (super)manifold. We will restrict our attention to just genuine manifolds in this paper and will not deal with supermanifolds at all. The reason for this is rooted in our applications of graded manifolds to mechanics rather than any fundamental geometric reasons.\smallskip

An important class of such manifolds are those that carry non-negative grading. We will furthermore require that this grading is associated with a smooth action $h:\R\ti F\to F$ of the monoid $(\R,\cdot)$ of multiplicative reals on a manifold $F$, a \emph{homogeneity structure} in the terminology of \cite{Grabowski2012}.
This action reduced to $\R_{>0}$ is the one-parameter group of diffeomorphism integrating the weight vector field, thus the weight vector field is in this case \emph{h-complete} \cite{Grabowski2013} and only \emph{non-negative integer weights} are allowed, so the algebra $\mathcal{A}(F)\subset C^\infty(F)$ spanned by homogeneous function is $\mathcal{A}(F) = \bigoplus_{i \in \mathbb{N}}\mathcal{A}^{i}(F)$.  This algebra  is referred to as the \emph{algebra of polynomial functions} on $F$. Importantly, we have  that for $t \neq 0$ the action $h_{t}$ is a diffeomorphism of $F$ and when $t=0$ it is a smooth surjection $\tau=h_0$ onto $F_{0}=M$, with the fibres being diffeomorphic to $\mathbb{R}^{N}$ (c.f.  \cite{Grabowski2012}).  Thus, the objects obtained are particular kinds of \emph{polynomial bundles} $\tau:F\to M$, i.e. fibrations which locally look like $U\times\R^N$ and the change of coordinates (for a certain choice of an atlas) are polynomial in $\R^N$. For this reason graded manifolds with non-negative weights and h-complete weight vector fields are also known as \emph{graded bundles} \cite{Grabowski2012}. Furthermore, the h-completeness condition  implies that  graded bundles are determined by the algebra of homogeneous functions on them.\smallskip

Canonical examples of graded bundles are, for instance, vector bundles, $n$-tuple vector bundles, higher tangent bundles $\sT^kM$, and multivector bundles $\we^n\sT E$ of vector bundles $\tau:E\to M$ with respect to the projection $\we^n\sT\tau:\we^n\sT E\to \we^n\sT M$ (see \cite{Grabowska2014}).  If the weight is constrained to be either zero or one, then the weight vector field is precisely a  vector bundle structure on $F$ and will be generally referred  to as an \emph{Euler vector field}.
\smallskip

We will regularly employ local coordinate systems adapted to the graded structure on a graded bundle. One can always pick an affine atlas of $F$ consisting of charts for which we  have homogeneous local coordinates $(x^{A}, y_{w}^{a})$, where $\w(x^{A}) =0$ and  $\w(y_{w}^{a}) = w$ with $1\leq w \leq k$, for some $k \in \mathbb{N}$ known as the \emph{degree}  of the graded bundle. It will be convenient to group all the coordinates with non-zero weight together.  The index  $a$ should be considered as a ``generalised index" running over all the possible weights. The label $w$ in this respect largely redundant, but it will come in very useful when checking the validity of various expressions.  The local changes of coordinates  are of the form

\begin{eqnarray}\label{eqn:translaws}
x^{A'} &=& x^{A'}(x),\\
\nonumber y^{a'}_{w} &=& y^{b}_{w} T_{b}^{\:\: a'}(x) + \sum_{\stackrel{1<n  }{w_{1} + \cdots + w_{n} = w}} \frac{1}{n!}y^{b_{1}}_{w_{1}} \cdots y^{b_{n}}_{w_{n}}T_{b_{n} \cdots b_{1}}^{\:\:\: \:\:\:\:\:a'}(x),
\end{eqnarray}
where $T_{b}^{\:\: a'}$ are invertible and the $T_{b_{n} \cdots b_{1}}^{\:\:\: \:\:\:\:\:a'}$ are symmetric in lower indices.

\smallskip
A graded bundle  of degree $k$  admits a sequence of  polynomial fibrations, where a point of $F_{l}$ is a class of the points of $F$ described  in an affine coordinate system by the coordinates of weight $\leq l$, with the obvious tower of  surjections

\begin{equation}\label{eqn:fibrations}
F=F_{k} \stackrel{\tau^{k}}{\longrightarrow} F_{k-1} \stackrel{\tau^{k-1}}{\longrightarrow}   \cdots \stackrel{\tau^{3}}{\longrightarrow} F_{2} \stackrel{\tau^{2}}{\longrightarrow}F_{1} \stackrel{\tau^{1}}{\longrightarrow} F_{0} = M,
\end{equation}

\noindent where the coordinates on $M$ have zero weight. Note that  $F_{1} \rightarrow M$ is a linear fibration and the other fibrations $F_{l} \rightarrow F_{l-1}$ are affine fibrations in the sense that the changes of local coordinates for the fibres are linear plus and additional additive terms of appropriate weight (c.f. \ref{eqn:translaws}).  The model fibres here are $\mathbb{R}^{N}$ (c.f.  \cite{Grabowski2012}).  We will also on occasion employ the projections $\tau^{i}_{j} : F_{i} \rightarrow F_{j}$, which are defined as the appropriate composition of the above projections. We will also use on occasion $\tau := \tau^{k}_{0} : F_{k} \rightarrow M$.\smallskip

There is also a ``dual"  sequence of submanifolds and their inclusions

\begin{equation}\label{eqn:submanifolds}
M := F_{0}= F^{[k]}  \hookrightarrow F^{[k-1]} \hookrightarrow \cdots \hookrightarrow F^{[0]} = F_{k},
\end{equation}

\noindent where we define, locally but correctly,

\begin{equation*}
F^{[i]} := \left\{\left. p \in F_{k} \right| y_{w}^{a} = 0 \:\:\: \textnormal{if} \:\: w \leq i   \right\}.
\end{equation*}

\noindent If we define $\mathcal{A}_{l}(F)$ to be the subalgebra of $\mathcal{A}(F)$ locally generated by functions of weight $\leq l$, then corresponding to the sequence of submanifolds and their inclusions is the filtration of algebras $C^{\infty}(M) = \mathcal{A}_{0}(F) \subset \mathcal{A}_{1}(F) \subset \cdots \subset \mathcal{A}_{k}(F) = \mathcal{A}(F)$.\smallskip

\noindent It will also be useful to consider the submanifold $\bar{F}_{l}:= F_{l}^{[l-1]}$ which is in fact linearly fibred   over $M$, with the linear coordinates carrying weight $l$.  The module of sections of $\bar{F}_{l}$ is identified with the $C^{\infty}(M)$-module $\mathcal{A}^{l}(F)$.

\smallskip
Morphisms between graded bundles are necessarily polynomial in the non-zero weight coordinates and respect the weight. Such morphisms can be characterised by the fact that they intertwine the respective weight vector fields  \cite{Grabowski2012}.\smallskip

The notion of a double vector bundle \cite{Pradines:1974} (or a higher n-tuple vector bundle)   is conceptually  clear in the graded language in terms of mutually commuting weight vector fields; see  \cite{Grabowski:2006,Grabowski2012}. This leads to the higher analogues known as \emph{n-tuple graded bundles}, which are  manifolds for which  the structure sheaf carries an $\mathbb{N}^{n}$-grading such that all the weight vector fields are h-complete and pairwise commuting.  In particular a \emph{double graded bundle} consists of a manifold and a pair of mutually commuting weight vector fields. If all the weights are either zero or one then we speak of an \emph{n-tuple vector bundle}.\smallskip

Let $\mathcal{M}$ be an n-tuple graded manifold and $(\Delta^{1}_{\mathcal{M}} , \cdots \Delta^{n}_{\mathcal{M}})$ be the collection of pairwise commuting weight vector fields. The local triviality of n-tuple graded bundles  means that we can always equip an n-tuple graded bundle with  an atlas such that the charts consist of coordinates that are simultaneously homogeneous with respect to the  weights associated with each weight vector field (c.f. \cite{Grabowski2012}). Thus, we can always equip an n-tuple graded bundle with homogeneous local coordinates $(y)$ such that $\underline{\w}(y) = (\w_{1}(y) , \cdots \w_{n}(y) ) \in \mathbb{N}^{n}$. The changes of local coordinates must respect the weights. Similarly, morphisms between n-tuple graded bundles respect the weights of the local coordinates. \smallskip

Each $\Delta^{s}_{\mathcal{M}}$  $(1 \leq s \leq n)$ defines a submanifold $\mathcal{M}_{s} \subset \mathcal{M}$ for which $\Delta^{s}_{\mathcal{M}}=0$. Moreover, as the weight vector fields are h-complete, we have a bundle structure, that is a surjective submersion $h_{0}^{s}:\mathcal{M} \rightarrow \mathcal{M}_{s}$ defined by the homogeneity structures associated with  each $\Delta^{s}_{\mathcal{M}}$.  As the homogeneity structures (or equivalently the weight vector fields) all pairwise commute, we have a whole diagram of fibrations  $h_{0}^{i_{1}}\circ h_{0}^{i_{2}}\circ \cdots \circ h_{0}^{i_{n}}: \mathcal{M} \rightarrow \mathcal{M}_{i_{1}}\cap \mathcal{M}_{i_{2}} \cap \cdots \cap \mathcal{M}_{i_{n}}$, where the fibres are homogeneity spaces.
\smallskip

 To set some useful notation, if $\Delta_{\mathcal{M}}$ is a weight vector field on $\mathcal{M}$, then we denote by $\mathcal{M}[\Delta_{\mathcal{M}} \leq l]$ the base manifold of the locally trivial fibration defined by taking the weight $>l$ coordinates with respect to this complimentary  weight to be the fibre coordinates. We have a natural projection that we will denote as

\begin{equation*}
\textnormal{p}^{\mathcal{M}}_{{\mathcal{M}}[\Delta_{\mathcal{M}} \leq l]}: \mathcal{M} \rightarrow \mathcal{M}[\Delta_{\mathcal{M}} \leq l].
\end{equation*}

Consider $\underline{l}$ and $\underline{m} \in \mathbb{N}^{n}$ we can define a partial ordering as $\underline{l} \preceq \underline{m} \Leftrightarrow \forall j \:\: l_{j} \leq m_{j}$. Then  given $\mathcal{M}$ above we  define

\begin{equation}
\mathcal{M}^{[\: \underline{m}\:  ]} :=  \left\{ \left. p \in \mathcal{M} \right| y = 0 \:\: \textnormal{if} \:\: \underline{\w}(y) \neq \underline{0} \: \: \textnormal{and} \:\: \underline{\w}(y) \preceq \underline{m} \right\}.
\end{equation}

As the changes of local coordinates must respect each of the weights independently and the fact that all the weights are non-negative  it is clear that  $\mathcal{M}^{[\underline{m} ]}$  are well-defined n-tuple graded subbundles of $\mathcal{M}$.
\smallskip

One can pass from an n-tuple graded bundle to an (n-r)-tuple graded bundle via the process of taking linear combinations of weight vector fields.  In particular one can construct the \emph{total weight vector field} as the sum of all the weight vector fields on an n-tuple graded bundle. However, note that the passage from an n-tuple graded bundle to a graded bundle via collapsing the weights is far from unique.

\begin{definition} A double graded bundle
$(\mathcal{M}_{k}, \Delta^{1}, \Delta^{2})$ such that
$\zD^1$ is of degree $k-1$ and $\zD^2$ is of degree 1, i.e. $\Delta^{2}$ is an Euler vector field, will be referred to as a  \emph{graded--linear bundle} of degree $k$, or for short a $\mathcal{GL}$-bundle.
\end{definition}
\noindent
It follows that $\mathcal{M}_{k}$ is of total weight $k$ with respect to $\Delta := \Delta^{1} + \Delta^{2}$ and a vector bundle structure
\begin{equation*}
\textnormal{p}^{\mathcal{M}_{k}}_{B_{k-1}}: \mathcal{M}_{k} \rightarrow B_{k-1}.
\end{equation*}
with respect to the projection onto the submanifold $B_{k-1} := \mathcal{M}_{k}[\Delta^{2} \leq 0]$ which inherits a graded bundle structure of degree $k-1$.

\medskip

\noindent \textbf{The tangent and phase lift:} Consider an arbitrary graded bundle $F_{k}$ of degree $k$. Let us employ  homogeneous coordinates $(x^{A}, y_{w}^{a})$ with $1 \leq w \leq k$. The weight vector field is locally given by
\begin{equation*}
\Delta_{F} := \sum_{w} w y_{w}^{a}\frac{\partial}{\partial y_{w}^{a}}.
\end{equation*}
The tangent bundle $\sT F_{k}$ naturally has the structure of a double graded bundle. To see this we employ homogeneous coordinates  with respect to the bi-weight which consists of the natural lifted weight and the weight due the vector bundle structure of the tangent bundle \cite{Grabowski:1995};

\begin{equation*}
 (\underbrace{x^{A}}_{(0,0)}, \hspace{5pt} \underbrace{y^{a}_{w}}_{(w,0)},\hspace{5pt} \underbrace{\dot{x}_{1}^{B}}_{(0,1)}, \hspace{5pt}  \underbrace{\dot{y}^{b}_{w+1}}_{(w,1)}).
 \end{equation*}
 In other words, we have the first weight vector field being simply the \emph{tangent lift} \cite{Grabowski:1995} of the weight vector field in $F_{k}$ and the second being the natural Euler  vector field on a tangent bundle;
 \begin{eqnarray}
 \Delta^{1}_{\sT F} &=& \D_{\sT}\Delta_{F} = \sum_{w} w y_{w}^{a}\frac{\partial}{\partial y_{w}^{a}} + \sum_{w} w \dot{y}_{w+1}^{a}\frac{\partial}{\partial \dot{y}_{w+1}^{a}},\\
 \nonumber \Delta^{2}_{\sT F} &=& \dot{x}_{1}^{A}\frac{\partial}{\partial \dot{x}_{1}^{A}} + \sum_{w} \dot{y}_{w+1}^{a}\frac{\partial}{\partial \dot{y}_{w+1}^{a}}.
 \end{eqnarray}
\medskip

The tangent bundle $\sT F_{k}$ of a graded bundle of degree $k$ can also be naturally considered as a graded bundle of degree $k+1$ via the total weight which is described by the total weight vector field
\begin{equation*}
\Delta_{\sT F_{k}} := \sum_{w} w y_{w}^{a}\frac{\partial}{\partial y_{w}^{a}} + \dot{x}_{1}^{A}\frac{\partial}{\partial \dot{x}_{1}^{A}} + \sum_{w} (w+1)\dot{y}^{a}_{w+1} \frac{\partial}{\partial \dot{y}_{w+1}^{a}}.
\end{equation*}

For example, it is easy to see that with respect to the standard bi-weight on $\sT F_{k}$ we have $\sT F_{k}^{[(0,1)]} = VF_{k}$, that is the vertical bundle with respect  to the projection $\tau: F_{k} \rightarrow M$. The weight vector fields on the vertical bundle are simply the appropriate restrictions of those on the tangent bundle.  Via passing to the total weight we see that $VF_{k}$ can be considered as a graded bundle of degree $k+1$. However, it will be prudent to shift the first component of bi-weight to allow us to consider the vertical bundle as a graded bundle of degree $k$.  In terms of the weight vector fields this is simply the redefinition $\Delta_{VF}^{1} \mapsto \Delta_{VF}^{1}- \Delta_{VF}^{2}$ and $\Delta_{VF}^{2} \mapsto \Delta_{VF}^{2}$ and then considering the redefined total weight.  We will use this assignment of weights for all vertical bundles in employed in this paper.
\smallskip

Similarly to the case of the tangent bundle, the cotangent bundle $\sT^{*}F_{k}$ of a graded bundle also comes naturally with the structure of a double graded bundle. However,  simply using the naturally induced weight would mean that the ``momenta" will have a negative component of their bi-degree and so the associated weight vector field cannot be complete; we would not remain in the category of graded bundles. Instead one needs to consider a phase lift of the weight vector field on $F_{k}$  (c.f. \cite{Grabowski2013}).  The \emph{k-phase lift} of the weight vector field $\Delta_{F}$ essentially  produces a shift in the induced bi-weight to ensure that everything is non-negative: it amounts  to the following choice of homogeneous coordinates
 \begin{equation*}
 (\underbrace{x^{A}}_{(0,0)}, \hspace{5pt} \underbrace{y^{a}_{w}}_{(w,0)},\hspace{5pt} \underbrace{\pi^{k+1}_{B}}_{(k,1)}, \hspace{5pt}  \underbrace{\pi_{b}^{k-w+1}}_{(k-w,1)}),
 \end{equation*}
 \noindent and the weight vector fields are given by
 \begin{eqnarray}
 \Delta^{1}_{\sT^{*}F} &=&  \sum_{w} w y_{w}^{a}\frac{\partial}{\partial y_{w}^{a}} + k \pi_{A}^{k+1} \frac{\partial }{\partial \pi_{A}^{k+1}} + \sum_{w} (k-w) \pi_{a}^{k-w+1}\frac{\partial}{\partial \pi_{a}^{k-w+1}},\\
 \nonumber \Delta^{2}_{\sT^{*}F} &=& \pi_{A}^{k+1}\frac{\partial}{\partial \pi^{k+1}_{A}} + \sum_{w} \pi^{k-w+1}_{a}\frac{\partial}{\partial \pi^{k-w+1}_{a}}.
 \end{eqnarray}
\medskip

Throughout this paper we will equip the cotangent bundles of graded bundle with  the k-phase lifted weight vector field, to ensure do not leave the category of n-tuple graded bundles.
\smallskip

The cotangent bundle $\sT^{*} F_{k}$ of a graded bundle of degree $k$ can be considered as a graded bundle of degree $k+1$ via the total weight which is described by the total weight vector field
\begin{equation*}
 \Delta_{\sT^{*}F_{k}} := \sum_{w} w y_{w}^{a}\frac{\partial}{\partial y_{w}^{a}} + (k+1) \pi_{A}^{k+1} \frac{\partial }{\partial \pi_{A}^{k+1}} + \sum_{w} (k-w+1) \pi_{a}^{k-w+1}\frac{\partial}{\partial \pi_{a}^{k-w+1}}.
 \end{equation*}

\subsection{The linearisation functor and linear duals}\label{sec:linearisation functor}

 An important concept as first uncovered in \cite{Bruce:2014} is the notion of the linearisation functor which takes a graded bundle and produces a double grade bundle for which the two side arrows are vector bundles.  The basic idea is to mimic the canonical embedding  $\sT^{k}M \subset \sT(\sT^{k-1}M)$ via ``holonomic vectors".\smallskip

\begin{definition}
The \emph{linearisation of a graded bundle} $F_{k}$ is the $\mathcal{GL}$-bundle defined as
\begin{equation*}
D(F_{k}) := VF_{k}[\Delta_{VF}^{1} \leq k-1],
\end{equation*}

so that we have the natural projection $\textnormal{p}^{VF_{k}}_{D(F_{k})} : VF_{k} \rightarrow D(F_{k})$.
\end{definition}

Let us briefly describe the local structure of the linearisation. Consider $F_{k}$ equipped with local coordinates $(x^{A}, y_{w}^{a}, z^{i}_{k})$, where the weights are assigned as $\w(x)=0$, $\w(y_{w}) = w$ $(1\leq w < k)$ and $\w(z) =k$. From this point on it will be convenient to single out the highest weight coordinates as well as the zero weight coordinates. In any natural homogeneous system on coordinates on the vertical bundle $VF_{k}$ one projects out the highest weight coordinates on $F_{k}$ to obtain $D(F_{k})$. One can easily check in local coordinates that doing so is well-defined. Thus on $D(F_{k})$ we have local homogeneous coordinates

\begin{equation}\label{coordinates:T*Fk}
 (\underbrace{x^{A}}_{(0,0)}, \hspace{5pt} \underbrace{y^{a}_{w}}_{(w,0)}; \hspace{5pt}  \underbrace{\dot{y}^{b}_{w}}_{(w-1,1)}, \hspace{5pt} \underbrace{\dot{z}_{k}^{i}}_{(k-1,1)}).
 \end{equation}

\noindent Note that with this assignment of the weights the linearisation of a graded bundle of degree $k$ is itself a graded bundle of degree $k$ when passing from the bi-weight to the total weight. It is important to note that the linearisation has the structure of a vector bundle $D(F_{k}) \rightarrow F_{k-1}$, hence the nomenclature ``linearisation".  The vector bundle structure is clear from the construction by examining the bi-weight. In \cite{Bruce:2014} it was  shown that  there is a canonical weight preserving  embedding of the graded bundle $F_{k}$

\begin{equation*}
\iota : F_{k} \hookrightarrow D(F_{k}),
\end{equation*}

 \noindent with respect to  the total weight on $D(F_{k})$, given by the image of the weight vector field $\Delta_{F} \in \Vect(F)$ considered as a geometric section of $VF_{k}$. In short  we have the following commutative diagram

\begin{diagram}[htriangleheight=20pt ]
 F_{k}& \rTo^{\Delta_{F}} & VF_{k} \\
 & \rdTo_{\iota}  & \dTo_{\textnormal{p}^{VF_{k}}_{D(F_{k})}}\\
  & & D(F_{k})
\end{diagram}
\smallskip

\noindent In natural local coordinates the nontrivial part of the embedding  is given by
\begin{equation*}
\iota^{*}(\dot{y}^{a}_{w}) = w \: y^{a}_{w}, \hspace{25pt} \iota^{*}(\dot{z}_{k}^{i}) = k\: z_{k}^{i}.
\end{equation*}

\noindent Elements of $\iota(F_{k})$ we refer to as \emph{holonomic vectors} in $D(F_{k})$.

\begin{theorem}\label{theom:functorial linearisation}
Linearisation is a functor from the category of graded bundles to the category of double graded bundles. In particular, the tower of fibrations (\ref{eqn:fibrations}) induces fibrations

\begin{equation}\label{eqn:linearized fibrations}
D(F_{k}) \stackrel{D(\tau^{k})}{\longrightarrow} D(F_{k-1})
\stackrel{D(\tau^{k-1})}{\longrightarrow}   \cdots \stackrel{D(\tau^{3})}{\longrightarrow} D(F_{2}) \stackrel{D(\tau^{2})}{\longrightarrow}D(F_{1})=F_1 \stackrel{\zt^1}{\longrightarrow} D(F_{0}) = M.
\end{equation}
\end{theorem}

\begin{example} We will show that $D(\sT^{k+1}M)\simeq\sT\sT^k M$. Let us recall first the well-known fact that iterated tangent bundles $\sT^l\sT^k M$ can be defined as equivalence classes of smooth maps $\chi:\R^2\rightarrow M$. The equivalence class of $\chi$ in $\sT^l\sT^k M$ will be denoted by $\st^{(l,k)}\chi$. For example,
if $l=1$, then in the equivalence class of $\sT\sT^kM\ni v=(q,\dot q, \cdots, q^{(k)}, \delta q,\delta \dot q,\cdots, \delta q^{(k)})$ there is a representative of the form
\begin{equation}\label{l1}
\chi(s,t)=q+s\delta q+t(\dot q+s\delta\dot q)+\frac12t^2(\ddot q+s\delta \ddot q)+\cdots+\frac{1}{k!}t^k(q^{(k)}+s\delta q^{(k)})+ t^{k+1}f(s,t)
\end{equation}
where $f$ is any smooth function of two variables.

Our goal is to define a lift of a tangent vector $v\in\sT_a\sT^kM$ to the vertical vector $v^{\sv}\in\sV_b\sT^{k+1} M$ provided $a=\zt^{k+1}_{k}(b)$. Let us take a representative $\chi$ of $v$ such that the curve $t\mapsto \chi(0,t)$ is a representative of $b$. In terms of coordinates it means that, if $b=(q,\dot q, \cdots, q^{(k)},q^{(k+1)} )$, then the function $f$ is of the form $\frac1{(k+1)!}q^{(k+1)}+tg(s,t)$. We define $v^\sv$ as an equivalence class of $\tilde\chi$, where $\tilde\chi(s,t)=\chi(st, t)$,
i.e.
$$v^{\sv}=\st^{(1,k+1)}\tilde\chi.$$
The vector $v^{\sv}$ is indeed vertical, since the curve $s\mapsto \tilde\chi(s,0)$ is a constant curve.
In coordinates, using (\ref{l1}),
\begin{multline}\label{l2}
\tilde\chi(s,t)=q+st\delta q+t(\dot q+st\delta\dot q)+\frac12t^2(\ddot q+st\delta \ddot q)+\cdots+ \\
\frac{1}{k!}t^k(q^{(k)}+st\delta q^{(k)})+ \frac1{(k+1)!}t^{k+1}q^{(k+1)}+t^{k+2}g(st,t)\,.
\end{multline}
After some reordering we get
\begin{multline}\label{l3}
\tilde\chi(s,t)=q+t(\dot q+s\delta q)+\frac12t^2(\ddot q+2s\delta \dot q)+\cdots+\frac{1}{k!}t^k(q^{(k)}+ks\delta q^{(k-1)})+\\ \frac{1}{(k+1)!}t^{k+1}(q^{(k+1)}+(k+1)s\delta q^{(k)})+t^{k+2}g(st,t)\,.
\end{multline}
In coordinates, from (\ref{l3}) we see that
$$v^\sv=(q,\,\dot q,\, \cdots,\, q^{(k)},\,q^{(k+1)},\, 0,\, \delta q,\,2\delta \dot q,\,\cdots,\, (k+1)\delta q^{(k)}).$$
The above map $\sT_a\sT^kM\rightarrow\sV_b\sT^{k+1} M$, $v\mapsto v^{\sv}$ is a linear isomorphism. It provides the identification between $D(\sT^{k+1}M)$ and $\sT\sT^k M$.
\end{example}

\begin{definition}
The \emph{linear dual of a graded bundle} $F_{k}$ is the dual of the vector bundle $D(F_{k}) \rightarrow F_{k-1}$, and we will denote this $D^{*}(F_{k})$.
\end{definition}

\begin{proposition}  We have
\begin{equation*}
D^{*}(F_{k}) \simeq \sT^{*}F_{k}[\Delta^{1}_{\sT^{*}F_{k}} \leq k-1],
\end{equation*}
which gives the canonical projection $\sT^{*}F_{k} \rightarrow D^{*}(F_{k})$.
\end{proposition}
In simpler terms, we can employ homogeneous coordinates inherited from  the cotangent bundle as
  \begin{equation*}
 (\underbrace{x^{A}}_{(0,0)} , \hspace{5pt} \underbrace{y_{w}^{a}}_{(w,0)}; \hspace{5pt} \underbrace{ \pi_{b}^{k-w+1}}_{{(k-w, 1)}} , \underbrace{\pi_{i}^{1}}_{(0,1)})\,.
 \end{equation*}
Note that by passing to total weight the linear dual of a graded bundle of degree $k$ is itself a graded bundle of degree $k$.

\begin{example}
If $F_{k} := \sT^{k}M$, then  $D^{*}(\sT^{k}M) \simeq \sT^{*}(\sT^{k-1} M)$.
\end{example}

\begin{definition}
The \emph{Mironian} of a graded bundle is the double graded bundle defined as
\begin{equation*}
\textnormal{Mi}(F_{k}) := D^{*}(F_{k})[\Delta_{D^{*}(F)}^{1} + k\: \Delta_{D^{*}(F)}^{2} \leq k]
\end{equation*}
\end{definition}

\noindent It is not hard to see that the local coordinates on the Mironian inherited from the coordinates on the linear dual are

\begin{equation*}
 (\underbrace{x^{A}}_{(0,0)} , \hspace{5pt} \underbrace{y_{w}^{a}}_{(w,0)}; \hspace{5pt} \underbrace{\pi_{i}^{1}}_{(0,1)}),
 \end{equation*}

 \noindent and so the Mironian of $F_{k}$ has the structure of a vector bundle over $F_{k-1}$. Moreover, we have the identification $\textnormal{Mi}(F_{k}) \simeq F_{k-1} \times_{M} \bar{F}_{k}^{*}$.

\begin{example}
If $F_{k} =  \sT^{k}M$, then $\textnormal{Mi}(\sT^{k}M) =  \sT^{k-1} M \times_{M}\sT^{*}M $.
\end{example}

\begin{remark}
To our knowledge the above example first appeared in the works of Miron, see \cite{Miron:2003}. His motivation, much like ours, was to develop a good notion of higher order Hamiltonian mechanics and this requires some notion of a dual of $ \sT^{k}M$. Our more general notion of the Mironian of a graded bundle will similarly play a fundamental role in our formulation of higher order Hamiltonian mechanics.
\end{remark}

Amongst all the possible fibrations of $\sT^{*}F_{k}$ we will use the following double graded
structure
\be\label{double-mironian}\xymatrix{
\sT^\ast F_k\ar[rr]^{\sT^*\zt^k} \ar[d]_{\zp_{F_k}} && \textnormal{Mi}(F_{k})\ar[d]^{\sv^*(\zt^k)}\\
F_k\ar[rr]^{\zt^k} && F_{k-1}}
\ee
or, composing the projections onto $F_{k-1}$ with $\zt^{k-1}_0:F_{k-1}\to M$,
\be\label{double-mironian1}\xymatrix{
\sT^\ast F_k\ar[rr]^{\sT^*\zt^k} \ar[d]_{\zp_{F_k}} && \textnormal{Mi}(F_{k})\ar[d]^{\sv^*(\zt^k_0)}\\
F_k\ar[rr]^{\zt^k_0} && M}\,.
\ee
\subsection{Weighted algebroids for graded bundles}
We are now in a position to describe the main geometric object needed to define  mechanics on a graded bundle.  A weighted algebroid should be thought of a generalised algebroid \cite{Grabowski:1999,Grabowska:2006,Grabowska:2008} carrying additional weights. Again for proofs of the statements here the reader is urged to consult \cite{Bruce:2014}.

\begin{definition}\label{def:algebroid}
 A \emph{weighted algebroid  for} $F_{k}$ is a morphisms of triple graded bundles
 \begin{equation}\label{epsilon}
 \epsilon : \sT^{\ast} D(F_{k}) \rightarrow \sT D^{\ast}(F_{k})
 \end{equation}
  covering the identity on the double graded bundle $D^{\ast}(F_{k})$.  The \emph{anchor map} is the map $\rho : D(F_{k}) \rightarrow \sT F_{k-1}$ underlying the map $\epsilon$, see the diagram below. The anchor map induces a graded morphism  $\hat{\rho} := \rho \circ \iota : F_{k} \rightarrow \sT F_{k-1}$ called the \emph{anchor map on $F_k$} or the \emph{k-th anchor map}. A graded bundle $F_k$ equipped with a graded morphism $\hat{\rho}: F_{k} \rightarrow \sT F_{k-1}$ we will call an \emph{anchored graded bundle}. A weighted algebroid is this specified by the pair $(F_{k}, \epsilon)$.  In this case we will also say that $F_{k}$ \emph{carries the structure of a weighted algebroid}.
\end{definition}

\begin{diagram}[htriangleheight=20pt]
\sT^{*} D^{*}(F_{k}) & & \rTo^{\epsilon} & & \sT D^{*}(F_{k}) & & \\
& \rdTo & & & \vLine & \rdTo & \\
\dTo & & D(F_{k}) & \rTo^{\rho} & \HonV & & \sT F_{k-1} \\
& & \dTo & & \dTo & & \\
D^{*}(F_{k}) & \hLine & \VonH & \rTo^{id} & D^{*}(F_{k}) & & \dTo \\
& \rdTo & & & & \rdTo & \\
& & F_{k-1} & & \rTo^{id}& &F_{k-1} \\
\end{diagram}
\medskip

The morphism (\ref{epsilon}) is known to be associated with a 2-contravariant tensor field $\Lambda_{\epsilon}$ on $D^{*}(F_{k})$. If  $\Lambda_{\epsilon}$ associated with $\epsilon$  is a bi-vector field on $D^{*}(F_{k})$, we speak about a \emph{weighted skew algebroid for} $F_{k}$. If the bi-vector field is a Poisson structure then we have a \emph{weighted Lie algebroid} for $F_{k}$.\smallskip

We will from this point on focus on weighted skew/Lie algebroid structures as our examples will be based on these structures and they seem general enough to cover a wide range of potentially interesting situations. To unravel the local structure, let us on $\sT^{*}D(F_{k})$ employ natural homogeneous  local coordinates

\begin{equation}\label{coordinates}
 (\underbrace{x^{A}}_{(0,0,0)}, \hspace{5pt} \underbrace{y_{w}^{a}}_{(w,0,0)}, \hspace{5pt}  \underbrace{\dot{y}_{w}^{b}}_{(w-1,1,0)}, \hspace{5pt} \underbrace{\dot{z}_{k}^{i}}_{(k-1,1,0)} ;\hspace{5pt} \underbrace{p_{B}^{k+1}}_{(k-1,1,1)}, \hspace{5pt}\underbrace{p^{k-w+1}_{c}}_{(k-w-1,1,1)},\hspace{5pt} \underbrace{\pi^{k-w+1}_{d}}_{(k-w,0,1)}, \hspace{5pt} \underbrace{\pi_{j}^{1}}_{(0,0,1)} ).
 \end{equation}

\noindent Similarly, let us on $\sT D^{*}(F_{k})$  employ natural homogeneous coordinates

 \begin{equation}
 (\underbrace{x^{A}}_{(0,0,0)}, \hspace{5pt} \underbrace{y_{w}^{a}}_{(w,0,0)}, \hspace{5pt}  \underbrace{\pi^{k-w+1}_{b}}_{(k-w,0,1)}, \hspace{5pt} \underbrace{\pi_{i}^{1}}_{(0,0,1)} ;\hspace{5pt} \underbrace{\delta x^{B}_{1}}_{(0,1,0)}, \hspace{5pt}\underbrace{\delta y_{w+1}^{c}}_{(w,1,0)},\hspace{5pt} \underbrace{\delta \pi^{k-w+2}_{d}}_{(k-w,1,1)}, \hspace{5pt} \underbrace{\delta \pi^{2}_{j}}_{(0,1,1)} ).
 \end{equation}

\noindent The labels on the coordinates refers to the total weight.  Note that  both $\sT^{*}D(F_{k})$ and $\sT D^{*}(F_{k})$ are naturally triple graded bundles and can be considered as graded bundles of degree $k+1$ by passing  to the total weight. The map $\epsilon$ for a weighted skew/Lie algebroid for $F_{k}$ is in these homogeneous coordinates given by

\begin{eqnarray}\label{epsilon:coordinates}
(\delta x^{A}_{1})\circ \epsilon &=& \dot{y}_{1}^{b}P_{b}^{A},\\
\nonumber (\delta y^{a}_{w+1})\circ \epsilon  & = & \sum \frac{1}{l!} \dot{y}_{w'}^{c}  y_{w_{1}}^{b_{1}} \cdots y_{w_{l}}^{b_{l}} P^{a}_{b_{l}\cdots b_{1}  c} + \delta_{w+1}^{k} \dot{z}_{k}^{i}P_{i}^{a},\\
\nonumber (\delta \pi_{c}^{k-w+2})\circ \epsilon  &=& \delta_{w}^{1}\left(P_{c}^{A}p_{A}^{k+1}  + \dot{z}^{i}_{k} P_{c i}^{j} \pi_{j}^{1}\right)+ \sum \frac{1}{l!} y_{w_{1}}^{b_{1}} \cdots y_{w_{l}}^{b_{l}}P^{d}_{b_{l} \cdots b_{1}c}p_{d}^{k-w'+1}\\
\nonumber &+& \sum \frac{1}{l!} \dot{y}_{w''}^{a}  y_{w_{1}}^{b_{1}} \cdots y_{w_{l}}^{b_{l}}\widehat{P}^{ d}_{b_{l}\cdots b_{1}ac}\pi_{d}^{k-w'+1} + \sum \frac{1}{l!} \dot{y}_{w'}^{a}  y_{w_{1}}^{b_{1}} \cdots y_{w_{l}}^{b_{l}}P^{ i}_{b_{l}\cdots b_{1}a c}\pi^{1}_{i}\\
\nonumber (\delta \pi_{i}^{2})\circ \epsilon  &=&  P_{i}^{a}p_{a}^{2}  + \dot{y}^{a}_{1}P_{a i}^{j}\pi^{1}_{j},
\end{eqnarray}
\noindent  where the sum are over the appropriate weights.

{ We will adopt the  following system of weight symbols and corresponding systems of coordinate indices on $F_{k}$ in order to condense our notation;
$$\begin{array}{c|c}
\text{weight} & \text{index} \\ \hline
0\leq u\leq k-1 & \mu, \nu \\
1\leq U\leq k & I,J \\
0\leq W\leq k & \alpha,\beta \\
1\leq w\leq k-1 &  a,b \\
0 & A \\
k & i
\end{array}$$
\noindent By convention the weight symbol will refer to the total weight of the given coordinate.}

Thus, we employ  $X_{u}^{\mu} =(x^A,y_w^a)$ as coordinates on $F_{k-1}$,  and $Y_{U}^{I}= (\dot y^b_w,\dot z^i_k)$ as the linear coordinates on $D(F_k)$. Finally we shall employ adapted  homogeneous coordinates

$$(X_{u}^{\mu}, Y_{U}^{I}, P^{k+1-u}_{\nu},\zP^{k+1-U}_{J})
$$
on $\sT^*D(F_k)$. Similarly, on $\sT D^*(F_k)$ we employ adapted homogeneous  local coordinates

$$(X_{u}^{\mu},\zP^{U}_{I},\zd X^{\nu}_{u+1},\zd\zP^{U+1}_J)\,.$$

In these coordinates, the map $\ze$, being  the identity on $(X^{\mu}_{u},\zP^{U}_{I})$, can be written in the  compact form
\bea
\zd X^{\mu}_{U}\circ\ze&=&\zr[u]^{\mu}_I(X)Y^I_{U-u}\,,\\
\nonumber\zd\zP^{U+1}_J\circ\ze&=&\zr[u]^{\nu}_J(X)P^{U+1-u}_{\nu}+C[u]^K_{IJ}(X)Y^I_{U'}\zP^{U+1-U'-u}_K\,,
\eea
where $\zr[u]$ and $C[u]$ are homogeneous parts of the structure functions of degree $u=0,\dots,k-1$.

\begin{example}\label{eg:higher tangent epsilon}
Consider the k-th order tangent bundle of a manifold $\sT^{k}M$. As proved in \cite{Bruce:2014}, $\sT^{k}M$ carries the structure of a weighted Lie algebroid. Let us examine this structure explicitly.  First, as we have already seen, $D(\sT^{k}M) \simeq \sT(\sT^{k-1} M)$ and thus $D^{*}(\sT^{k}M) \simeq \sT^{*}(\sT^{k-1} M)$. Then the weighted Lie algebroid for $\sT^{k}M$ is a morphism between the triple graded bundles
\begin{equation*}
\sT^{*}\left( \sT(\sT^{k-1} M)\right) \stackrel{\epsilon}{\longrightarrow} \sT\left( \sT^{*}(\sT^{k-1} M)\right).
\end{equation*}
\noindent Let us now employ homogeneous $\{X_{u}^{\mu}, Y_{U}^{I} \}$  on $D(\sT^{k}M)$. The weighted Lie algebroid structure is then specified by

\begin{equation*}
\delta X^{\mu}_{U} \circ \epsilon = \delta_{I}^{\mu}\: Y_{U}^{I},  \hspace{25pt}\textnormal{and} \hspace{25pt} \delta \Pi^{U+1}_{J} \circ \epsilon = \delta_{J}^{\nu}\: P_{\nu}^{U+1},
\end{equation*}

\noindent in the notation introduced above.
\end{example}

\begin{example}\label{eg:Lie algebroid epsilon}
Let $\mathcal{G} \rightrightarrows M$ be an arbitrary Lie groupoid with source map $\underline{s}: \mathcal{G} \rightarrow M$ and target map $\underline{t}: \mathcal{G} \rightarrow M$. There is also the inclusion map $\iota : M \rightarrow \mathcal{G}$   and a partial multiplication  $(g,h) \mapsto g\cdot h$ which is defined on $\mathcal{G}^{(2)}=\{(g,h)\in\mathcal{G}\ti\mathcal{G}: \underline{s}(g) = \underline{t}(h)\}$. Moreover, the manifold $\mathcal{G}$ is foliated by $\underline{s}$-fibres $\mathcal{G}_{x}= \{ \left.g \in \mathcal{G}\right| \underline{s}(g) =x\}$, where $x \in M$. As by definition the source (and target) map is a submersion, the $\underline{s}$-fibres are themselves smooth manifolds. Geometric objects associated with this foliation will be given the superscript $\underline{s}$.  For example, the distribution tangent to the leaves of the foliation will be denoted by $\sT \mathcal{G}^{\underline{s}}$. Let us now consider the $k$th order version of this, that is the subbundle $\sT^{k}\mathcal{G}^{\underline{s}} \subset \sT^{k}\mathcal{G}$ consisting of all higher order velocities tangent to some $\underline{s}$-leaf $\mathcal{G}_{x}$. That is we identify $\sT^{k}\mathcal{G}^{\underline{s}} $ with the union of $\sT^{k}\mathcal{G}_{x}$  over all $\underline{s}$-leaves $\mathcal{G}_{x}$. The relevant graded bundle  (over $\iota(M) \simeq M $) here is (c.f. \cite{Jozwikowski:2014})
\begin{equation*}
 F_{k} = \textnormal{A}^{k}(\mathcal{G}) := \left. \sT^{k}\mathcal{G}^{\underline{s}}\right|_{\iota(M)},
 \end{equation*}
which of course inherits its graded  bundle structure as a substructure of $\sT^{k}\mathcal{G}$  with respect to the projection $\zt_{k}:\textnormal{A}^k(\mathcal{G})\to M$. The claim is that $ \textnormal{A}^{k}(\mathcal{G})$ canonically carries the structure of a weighted Lie algebroid, this was first proved in \cite{Bruce:2014}. To see this weighted Lie algebroid we first need the linearisation. In \cite{Bruce:2014} it was shown that the linearisation of $\textnormal{A}^{k}(\mathcal{G})$ is given by
\begin{equation}\label{prolongation}
D(\textnormal{A}^{k}(\mathcal{G}) ) \simeq \{ (Y,Z) \in \textnormal{A}(\mathcal{G}) \times_{M} \sT\textnormal{A}^{k-1}(\mathcal{G})|\quad \bar{\rho}(Y) = \sT \bar{\tau}(Z)  \}\,,
\end{equation}
viewed as a vector bundle over $\textnormal{A}^{k-1}(\mathcal{G})$ with respect to the obvious projection of part $Z$ onto ${\textnormal{A}^{k-1}(\mathcal{G})}$, where $\bar{\rho} : \textnormal{A}(\mathcal{G}) \rightarrow \sT M$ is the standard anchor of the Lie algebroid and $\bar{\tau}: \textnormal{A}^{k-1} (\mathcal{G}) \rightarrow M$ is the obvious projection.  Note that this object is a canonical example of a \emph{Lie algebroid prolongation} \cite{Carinena:2001,Popescu:2001,Martinez:2001}.

\smallskip
We shall pick local coordinates $X_{u}^{\mu}= (x^{A}, y^{a}_{w})$ on  and $\textnormal{A}^{k-1}(\mathcal{G})$ and the corresponding linear coordinates $Y_{U}^{I} = (\dot{y}_{1}^{b}, \dot{y}_{w+1}^{c})$ on the linearisation  $D(\textnormal{A}^{k}(\mathcal{G}) )$.
Then, with the above choice of coordinates, the natural adapted coordinates on $\sT^{*}D(\textnormal{A}^{k}(\mathcal{G}))$  are  $\left(X_{u}^{\mu} , Y_{U}^{I} , P_{\nu}^{k-1 -u}, \Pi_{J}^{k-1-U} \right)$ and similarly let us employ local coordinates $\left( X_{u}^{\mu} , \Pi_{I}^{k-1-U}, \delta X_{U}^{\nu},  \delta \Pi_{J}^{U+1}\right)$ on $\sT D^{*}(\textnormal{A}^{k}(\mathcal{G}))$. The weighted Lie algebroid structure is given by
\bea
\zd X^{\mu}_{U}\circ\ze&=&\zr[0]^{\mu}_I(x)Y^I_{U}\,,\\
\nonumber\zd\zP^{U+1}_J\circ\ze&=&\zr[0]^{\nu}_J(x)P^{U+1}_{\nu}+ \delta_{U}^{k}\:C[0]^K_{IJ}(x)Y^I_{1}\zP^{U}_K\,,
\eea
\noindent where $\zr[0]_{I}^{\mu} = (\zr_{a}^{A}(x),\delta_{a}^{b} )$  and  $C[0]_{JI}^{K} =C_{ab}^{c}(x)$.  Here $(\zr_{a}^{A}(x), C_{ab}^{c}(x))$ are the structure functions of the Lie algebroid $\textnormal{A}(\mathcal{G})$.
\end{example}

\begin{example}\label{e-lalg}
A particular case of the above construction is the case of a Lie group $\mathcal{G}=G$. Then, we can identify $\textnormal{A}^k(G)$ with $\mathfrak{g}_{k}=\mathfrak{g}[1]\ti\dots\ti\mathfrak{g}[k]$, where $\mathfrak{g}$ is the Lie algebra of $G$, and
$D(\textnormal{A}^k(G))$ with $\mathfrak{g}\ti\sT\mathfrak{g}_{k-1}$ viewed as a vector bundle over $\mathfrak{g}_{k-1}$.  Here, $\mathfrak{g}[i]$ is the space $\mathfrak{g}$ with $i$th shift in the grading, so that linear functions on $\mathfrak{g}$ have weight $i$.
\end{example}
\begin{example}\label{e-pair}
Another particular case is the case of a pair groupoid $\mathcal{G}=M\ti M$. It is easy to see that $\textnormal{A}^k(M\ti M)=\sT^k M$.
\end{example}

\begin{remark} The weighted Lie algebroid  $D(\textnormal{A}^2(\mathcal{G}))$, without reference to any graded structure, first appeared in the works of  Mart\'{\i}nez \cite{Martinez:2001} and  Cari\~{n}ena \&  Mart\'{\i}nez \cite{Carinena:2001}, which they referred to as the \emph{prolongation of a Lie algebroid}.   A general prolongation has also been considered by Popescu \& Popescu \cite{Popescu:2001}. The motivation for these works was to understand geometric mechanics on Lie algebroids. Here we see that the prolongation of a Lie algebroid naturally appears in the context of weighted algebroids as does the `higher order' prolongations.
\end{remark}

Actually, as was shown in \cite{Jozwikowski:2014}, the graded bundle $\textnormal{A}^k(\mathcal{G})$ is completely determined by the Lie algebroid $\textnormal{A}=\textnormal{A}(\mathcal{G})$, or better to say, by the anchored bundle structure of $\zt:\textnormal{A}\to M$. Moreover, in  \cite{Jozwikowski:2014}  the following iterative procedure to construct $\textnormal{A}^k$ for an anchored bundle was presented:
\bea A^2&=&\{ Z\in\sT A\ |\ \bar\zr\circ\zt(Z)=\sT\zt(Z)\}\,,\\
A^{k+1}&=&\sT A^k\cap\sT^k A\quad\text{for}\quad k\ge2\,.\label{inductive}
\eea
In the latter condition we clearly understand $\sT A^k$ (inductively) and $\sT^k A$ as subsets of $\sT\sT^{k-1}A$. 

It is easy to see (c.f. \cite[Section 2]{Grabowska:2008}) that $A^2$ is the subbundle of $\sT A$ of first jets (tangent prolongations) of \emph{admissible curves} in the anchored bundle $A$, i.e. curves satisfying
$$\bar\zr\circ \zg=\st(\zt\circ\zg)\,,$$
where $\st$ denotes the tangent prolongation. In \cite{Grabowska:2008}, $A^2$ was denoted $\sT^{hol}A$ and its elements were called \emph{holonomic vectors}. The set of $A$-holonomic vectors $\sT^{hol}A$ can be equivalently characterized as the subset in $\sT A$ which is mapped {\it via} $\sT\bar\zr:\sT A\ra\sT\sT M$ to classical holonomic vectors $\sT^2 M$, that justifies this name. In other words,
$$A^2=(\sT\bar\zr)^{-1}(\sT^2 M)\,.$$
The inductive definition gives easily the following.
\begin{theorem}
If $A$ is an anchored bundle, then $A^k$, $k\ge 1$, is the bundle of $(k-1)$-jets of admissible curves in $A$. Alternatively,
\be\label{holonomic}  A^{k}= (\sT^{k-1}\bar\zr)^{-1}(\sT^k M)\,,\ee
where we view $\sT^{k-1}\bar\zr$ as the map
\be\label{higher-rho}\sT^{k-1}\bar\zr:\sT^{k-1}A\to\sT^{k-1}\sT M\simeq\sT\sT^{k-1}M\,.
\ee
\end{theorem}
\begin{proof} Indeed, the inductive characterization of $A^{k+1}$ (\ref{inductive}) tells us that the elements of $A^{k+1}$ are jets of those $k$th tangent prolongations of curves in $A$ which project onto $(k-1)$th jets of admissible curves, thus are jets of admissible curves.

The inductive proof of (\ref{holonomic}) follows easily from the fact, for $k\ge 2$, that the map (\ref{higher-rho}) is a part of
$$\sT\sT^{k-2}\bar\zr:\sT\sT^{k-2}A\to\sT\left(\sT\sT^{k-2}M\right)$$
for which $\sT A^{k-1}$ can be characterized as the inverse image of $\sT\sT^{k-1}M$. Hence,
$$\sT^kA\cap\sT A^{k-1}=\sT^kA\cap\left((\sT\sT^{k-2}\bar\zr)^{-1}(\sT\sT^{k-1}M)\right)=
(\sT^{k-1}\bar\zr)^{-1}(\sT^k M)\,.$$
\end{proof}

Note we have made no reference to any Lie groupoid structure in this construction. Having defined $A^k$, we can use (\ref{prolongation}) to obtain $D(A^k)$ (just forgetting $\mathcal{G}$).

The situation where we have some additional structure on $\textnormal{A}$ is of course more interesting. If we suppose that $\textnormal{A}$ is in fact a Lie algebroid, not necessarily integrable, then the $\mathcal{GL}$-bundle $D(\textnormal{A}^{k})$ is a weighted Lie algebroid of degree $k$, where the brackets are defined using projectable sections, see for example \cite{Cortes:2006}. The space of projectable sections is closed under this Lie bracket; this follows from
\be\label{D}\bar\zr([X,Y])=[\bar\zr(X),\bar\zr(Y)]\,,\ee
and the fact that $\bar{\tau}$ is a projection.  Because we can chose a local basis consisting of projectable sections this Lie bracket extends to all sections.\smallskip

Note that we do not have explicit reference to the Jacobi identity in defining this Lie bracket, the only essential piece for consistency is the compatibility  of the anchor with the brackets (\ref{D}), that is, we only require the underlying structure to be that of an \emph{almost Lie algebroid} in the terminology of \cite{Grabowski:2011}. The condition (\ref{D}) has appeared as a necessary condition for the possibility of developing an appropriate variational calculus already in \cite{Grabowska:2008} (see also \cite{Grabowski:2011,Jozwikowski:2014}).

The above observations lead to the following theorem.
\begin{theorem}\label{thrm:almost Lie}
For any almost Lie algebroid $(\textnormal{A}, [,] , \bar{\rho})$,  the $\mathcal{GL}$-bundle $D(\textnormal{A}^{k})$  comes equipped with the structure of a weighted skew algebroid whose construction is outlined above.
\end{theorem}

For further examples and in particular slightly more general weighted algebroids, that is where we can relax the underlying double graded structure not to necessarily be associated with the linearisation of a graded bundle, see \cite{Bruce:2014}.


\section{Phase spaces, mechanics on algebroids and constraints}\label{sec:PhaseSpacesMechanicsConstraints}
We will later make use of the affine structure of graded bundles, geometric mechanics on Lie algebroids and vakononic constraints. In order for this paper to be relatively self-contained we review what we will need in the later sections of this paper.

\subsection{Affine phase spaces}\label{sec:Affine phase spaces}
Since we will make use of the affine structure of the fibrations $F_k\to F_{k-1}$, let us present
basic concepts  of the geometry of affine bundles.

Let $A$ be an affine space modelled on a vector space $\sv(A)$.
This means that the commutative group $\sv(A)$ acts freely and transitively on $A$ by addition
$$A\ti\sv(A)\ni (a,v)\mapsto a+v\,.$$
In other words, the naturally defined differences $a_1-a_2$ of points of $A$ belong to $\sv(A)$.
On affine spaces there are defined \emph{affine combinations} of points, $ta_1+(1-t)a_2$, for all $a_1,a_2\in A$ and $t\in\R$. Note that {\it convex combinations} are those affine combinations
$ta_1+(1-t)a_2$ for which $0\le t\le 1$.

All this can be extended to affine bundles $\zt:A\to N$ modelled on a vector bundle $\sv(\zt):\sv(A)\to N$. Any vector bundle is an affine bundle and fixing a section $a_0$ of $A$ induces an isomorphism of affine bundles $A$ and $\sv(A)$,
$$\sv_x(A)\ni v\mapsto a_0(x)+v\in A_x\,.$$
Using coordinates $(x^a)$ in the open set $\mathcal{O}\subset N$, a local section $a_0: \mathcal{O}\rightarrow A$,
and local base of sections $e_i:\mathcal{O}\rightarrow \sv(A)$, we can construct an adapted coordinate
system $(x^a, y^i)$ in $\tau^{-1}(\mathcal{O})$. An element $a\in A$ can be written as
$a=a_0(\tau(a))+y^ie_i(\tau(a))$.
\smallskip

For the affine space $A$, we consider its {\it affine dual}, i.e.
the space $A^\dag=\Aff(A,\R)$ of all affine maps from $A$ to $\R$.
\begin{definition}\label{def:av}(\cite{Grabowska:2004,Grabowska:2007})
An \emph{AV-bundle} (a \emph{bundle of affine values}) is an affine bundle $\zz:\sZ\to \cM$ modelled on a trivial one-dimensional vector bundle $\cM\ti \R$
\end{definition}
\noindent Sections of the AV-bundle $\zeta$ are regarded as affine analogs of functions on a manifold $\cM$.
The bundle $\zt^\dag:A^\dag \longrightarrow N$, where $A^\dag=\Aff(A,\R)$ is
the set of all affine maps on fibres of $\zt$, is called the \emph{affine dual bundle}. Instead of $A^\dag(\R)$ we will write simply $A^\dag$.

Every affine map $\phi:A_1\to A_2$ has a well-defined {\it linear part},  $\sv(\phi):\sv(A_1)\to\sv(A_2)$, therefore there is a projection
\begin{equation}\label{theta}\zz: A^\dag\longrightarrow \sv(A)^\ast.\end{equation}
The above bundle, denoted with $\AV(A^\dag)$,  is a canonical example of an AV-bundle which is modelled on
\begin{equation}\label{hb}
\sv(A)^\ast\times \R\,.
\end{equation}
Using the dual base sections
$\ze^i:\mathcal{O}\rightarrow \sv(A)^\ast$, we construct
an adapted coordinate system $(x^a, p_i, r)$ on $(\zt^\dag)^{-1}(\mathcal{O})$.
An affine map $\varphi$ on $A_q$ can be written as $\varphi(a)=p_i\ze^i(a-a_0(q))-r$, i.e. $r=-\varphi(a_0)$.
The map $\zz$ in coordinates reads $(x^a,p_i,r)\mapsto(x^a,p_i)$.\smallskip
\begin{remark}\label{r1} The choice of coordinate $r$ with the ``minus" sign may seem unnatural, but is well motivated (c.f. \cite{Grabowska:2004}). Assume for simplicity that the coordinates $x$ are not present, so we deal with just an affine space $A$.

Having coordinates $(p,r)$ on $A^\dag$ we want to identify (locally) sections of $\zz:\A^\dag=\Aff(A,\R)\to \sv^*(A)$ with functions on $\sv^*(A)$.
For instance, any $a\in A$ defines canonically a section $\zs_a$ of $\zz$ being the zero-level set of the function $f_a:\Aff(A,\R)\to\R$, $f_a(\zf)=\zf(a)$.

In our coordinates, $f_a(p,r)=\la p,a-a_0\ran-r$, so that the image of $\zs_a$ is $\{ (p,\la p, a-a_0\ran)\}$, so correctly the graph of $a-a_0\in\sv(A)$ as the linear function on $\sv^*(A)$. Generally, in the introduced canonical coordinates we will interpret functions $F:\sv^*(A)\to\R$ as sections $(x,p)\mapsto(x,p,-F(x,p))$ of the AV-bundle
$\zz$.
\end{remark}

As we have already mentioned, in many constructions functions on a manifold can be replaced by sections of an AV-bundle
over that manifold. We can obtain also an affine analog of the differential of a function
and an affine version of the cotangent bundle as follows. Given an AV-bundle $\zz:\sZ\to \cM$ and
$\sigma_1,\sigma_2\in\Sec(\sZ)$, $\sigma_1-\sigma_2$ may be seen as a map
$$\sigma_1-\sigma_2:\cM\to \R\,,$$
so the differential
$$\xd(\sigma_1-\sigma_2)(m)\in \sT^*_m\cM$$
is well defined.

\begin{definition}\label{def:phase}
The \emph{phase bundle} $\sP\sZ$ of an AV-bundle $\sZ$ is the affine bundle of cosets $\uxd \sigma(m)=[(m,\sigma)]$
(`affine differentials') of the equivalence relation
\begin{equation*}(m_1,\sigma_1)\sim(m_2,\sigma_2)\ \Leftrightarrow\ m_1=m_2\,,\quad \xd(\sigma_1-\sigma_2)(m_1)=0\,.\end{equation*}
\end{definition}

The projection $\sP\zz:\sP\sZ\to\cM$ makes $\sP\sZ$ into an affine bundle modelled on $\sT^*\cM$. Indeed, fixing a section $\sigma_0:\cM\to \sZ$, we get a diffeomorphism
\begin{equation*}\psi:\sP\sZ\to\sT^*\cM\,, \quad \uxd \sigma(m)\mapsto \xd(\sigma-\sigma_0)(m)\,.\end{equation*}
Moreover, as the canonical symplectic form on $\sT^*\cM$ is linear and invariant with respect to translations by closed 1-forms, its pull-back does not depend on the choice of $\sigma_0$, so
$\sP\sZ$ is canonically a symplectic manifold.\smallskip

Like in the case of the cotangent bundle $\sT^*\cM$, where the image $\xd\! f(\cM)$ of the differential of a function $f$ on $\cM$ yields a lagrangian submanifold, also the image $\uxd \sigma(\cM)$ of a section of $\sigma\cM\to\sZ$ of $\sZ$ is a lagrangian submanifold in $\sP\sZ$.

\begin{example} For an affine bundle $\zt:A\to N$, take $\Z=\AV(A^\dag)$. Then, $\sP(\AV(A^\dag))$ is an affine bundle over $\sv^*(A)$ which is canonically symplectic. We will denote $\sP(\AV(A^\dag))$, with some abuse of notation, with $\sP A^\dag$.
It is easy to see that $\sP A^\dag$ is also canonically a vector bundle over $A$. Both bundle structures make the affine phase bundle into a \emph{double affine bundle} \cite{Grabowska:2007,Grabowski:2010}. Actually, one affine structure is linear, so we can speak about a \emph{vector-affine bundle}. The situation is similar to that with the cotangent bundle $\sT^*A$ which, besides the vector fibration over $A$ is canonically an affine bundle over  $\sv^*(A)$. This is an affine analog of the well-known vector fibration of $\sT^*E$ over $E^*$, for a vector bundle $E$.
The local identification of $A^\dag$ with $\sv^\ast(A)$ by means of the local section $a_0$ yields the local identification of $\sP A^\dag$ with $\sT^\ast\sv^\ast(A)$. We can therefore use for $\sP A^\dag$ coordinates $(x^a, p_i, \eta_a, y^i)$ pulled back from $\sT^\ast\sv^\ast(A)$. The canonical symplectic form $\omega_{A^\dag}$ in coordinates reads
$\omega_{A^\dag}=\xd \eta_a\wedge \xd x^a+\xd y^i\wedge \xd p_i$.
\end{example}

We have the following analog of the isomorphism (\ref{iso}).
\begin{theorem}\label{th:1} There is a canonical isomorphism $\cR_A$ of vector-affine bundles
\be\label{affR}\xymatrix{
 & \sT^*A \ar[rrr]^{\cR_A} \ar[dr]^{\sT^*\zt}
 \ar[ddl]_{\zp_A}
 & & & \sP A^\dag\ar[dr]^{\sP^*\zt}\ar[ddl]_/-10pt/{\sP\zt}
 & \\
 & & \sv^*(A)\ar@{=}[rrr]\ar[ddl]_/-10pt/{\sv^*(\tau)}
 & & & \sv^*(A) \ar[ddl]_{\sv^*(\tau)}\\
 A\ar@{=}[rrr]\ar[dr]^{\zt}
 & & & A\ar[dr]^{\zt} & &  \\
 & M\ \ar@{=}[rrr]& & &\ M &
}
\ee
covering the identities over $A$ and $\sv^*(A)$, which is simultaneously an anti-symplectomorphism.
\end{theorem}

More specifically, if $A$ is the affine bundle $\zt^k:F_k\to F_{k-1}$, then $\sv^*(A)=F_{k-1} \times_{M}\bar{F}_{k}^*$ is the Mironian $\textnormal{Mi}(F_{k})$ of $F_k$.  In coordinates $(X^{\mu}_{u}, X^i_k, \zX_{\nu}^{k-u+1}, \Theta_i^1)$ in $\sT^\ast F_k$
and $(X^{\mu}_{u}, \Theta^1_i, \zG^{k+1-u}_{\mu}, X^i_k)$ in $\sP F_k^\dag$ the isomorphism $\cR_{F_k}$ reads

$$(X^{\mu}_{u}, \Theta^1_i, \zG^{k+1-u}_{\nu}, X^i_k)\circ \cR_{F_k}=(X^{\mu}_{u}, X^i_k, -\zG_{\nu}^{k-u+1}, \Theta_i^1) .$$

\begin{remark}
The above isomorphism is a particular case of the identification $\sP(\bA)\simeq\sP(\bA^\#)$
valid for any \emph{special affine bundle} $\bA$ \cite{Grabowska:2007}. Note also that the above diagram reduces to (\ref{double-mironian}) if as the affine bundle $\zt:A\to N$ we take $\zt^k:F_k\to F_{k-1}$.
\end{remark}


\subsection{Mechanics on algebroids}\label{sec:mechanics on algebroids}
We start this section by  introducing some notation. Let $M$ be a smooth manifold and let $(x^a), \ a=1,\dots,n$, be a coordinate system in $M$. We denote by
$\zt_M \colon \sT M \rightarrow M$ the tangent vector bundle and by $\zp_M \colon \sT^\* M\rightarrow M$ the
cotangent vector bundle. We have the induced (adapted) coordinate systems $(x^a, {\dot x}^b)$ in $\sT M$ and
$(x^a, p_b)$ in $\sT^\* M$.
        Let $\zt\colon E \rightarrow M$ be a vector bundle and let $\zp
\colon E^\* \rightarrow M$ be the dual bundle.
  Let $(e_1,\dots,e_m)$  be a basis of local sections of $\zt\colon
E\rightarrow M$ and let $(e^{1}_*,\dots, e^{m}_*)$ be the dual basis of local sections of $\zp\colon
E^\*\rightarrow M$. We have the induced coordinate systems:
    \beas
    (x^a, y^i),\quad & y^i=\zi(e^{i}_*), \quad \text{in} \ E,\\
    (x^a, \zx_i), \quad &\zx_i = \zi(e_i),\quad \text{in} \ E^\* ,
    \eeas
    where the linear functions  $\zi(e)$ are given by the canonical pairing
    $\zi(e)(v_x)=\la e(x),v_x\ran$. Thus we have local coordinates
    \beas
    (x^a, y^i,{\dot x}^b, {\dot y}^j ) &  \quad \text{in} \ \sT E,\\
    (x^a, \zx_i, {\dot x}^b, {\dot \zx}_j) & \quad \text{in} \ \sT E^\* ,\\
    (x^a, y^i, p_b, \zp_j) & \quad \text{in}\ \sT^\*E,\\
    (x^a, \zx_i, p_b, \zf^j) & \quad \text{in}\ \sT^\* E^\* .
    \eeas

The cotangent bundles $\sT^\*E$ and $\sT^\*E^\*$ are examples of
double vector bundles which are canonically isomorphic with the isomorphism
    \be\label{iso}\cR_E \colon \sT^\*E \longrightarrow \sT^\* E^\*
                    \ee
  being simultaneously an anti-symplectomorphism  (cf. \cite{Kolar:1997,Konieczna:1999,Grabowski:1999}). In local coordinates, $\cR_\zt$ is given by
    $$\cR_E(x^a, y^i, p_b, \zp_j) = (x^a, \zp_i, -p_b,y^j).
                              $$
This means that we can identify coordinates $\zp_j$ with $\zx_j$, coordinates $\zf^j$
with $y^j$, and use the coordinates $(x^a, y^i, p_b, \zx_j)$ in $\sT^\ast E$ and the
coordinates $(x^a, \zx_i, p_b,y^j)$ in $\sT^\ast E^\ast$, in full agreement with
(\ref{iso}).

For the standard concept and theory of {\it Lie algebroids} we refer to the survey
article \cite{Ma0} (see also \cite{Gr,Mac}). It is well known that Lie algebroid
structures on a vector bundle $E$ correspond to linear Poisson tensors on $E^\*$. A
2-contravariant tensor $\zP$ on $E^\*$ is called {\it linear} if the corresponding
mapping $\widetilde{\zP} \colon \sT^\* E^\* \rightarrow \sT E^\*$ induced by
contraction, $\wt{\zP}(\zn)=i_\zn\zP$, is a morphism of double vector bundles. One can
equivalently say that the corresponding bracket of functions is closed on (fiber-wise)
linear functions. The commutative diagram
\be\label{xx}\xymatrix{
\sT^\ast E^\ast\ar[r]^{\widetilde\Pi_\ze}  & \sT E^\ast \\
\sT^\ast E\ar[u]_{\cR_E}\ar[ur]^{\ze} & },
\ee
describes a one-to-one correspondence between linear 2-contravariant tensors $\zP_\ze$ on $E^\*$ and morphisms
$\ze$ (covering the identity on $E^\*$) of the following double vector bundles (cf. \cite{Grabowski:1997,Grabowski:1999,Grabowska:2006,Grabowska:2008}):

\be\xymatrix{
 & \sT^\ast E \ar[rrr]^{\varepsilon} \ar[dr]^{\pi_E}
 \ar[ddl]_{\sT^\ast\tau}
 & & & \sT E^\ast\ar[dr]^{\sT\pi}\ar[ddl]_/-20pt/{\tau_{E^\ast}}
 & \\
 & & E\ar[rrr]^/-20pt/{\zr}\ar[ddl]_/-20pt/{\tau}
 & & & \sT M \ar[ddl]_{\tau_M}\\
 E^\ast\ar[rrr]^/-20pt/{id}\ar[dr]^{\pi}
 & & & E^\ast\ar[dr]^{\pi} & &  \\
 & M\ar[rrr]^{id}& & & M &
}\label{F1.3}
\ee
In local coordinates, every  such $\ze$ is of the form
\be\label{F1.4}
\ze(x^a,y^i,p_b,\zx_j) = (x^a, \zx_i, \zr^b_k(x)y^k, c^k_{ij}(x) y^i\zx_k + \zs^a_j(x) p_a)
\ee
(summation convention is used) and it corresponds to the linear tensor $$ \zP_\ze =c^k_{ij}(x)\zx_k
\partial _{\zx_i}\otimes \partial _{\zx_j} + \zr^b_i(x) \partial _{\zx_i}
\otimes \partial _{x^b} - \zs^a_j(x)\partial _{x^a} \otimes \partial _{\zx_j}.
$$
The morphisms (\ref{F1.3}) of double vector bundles covering the identity on $E^\*$ has been called an {\it
algebroid} in \cite{Grabowski:1999}, while a {\it Lie algebroid} has turned out to be an algebroids for which the tensor
$\zP_\ze$ is a Poisson tensor. If the latter is only skew-symmetric, we deal with a \emph{skew algebroid}.

\medskip
Combining (\ref{F1.3}) with (\ref{xx}) we get the algebroid version of the \emph{Tulczyjew triple} in the form of the diagram:

$$
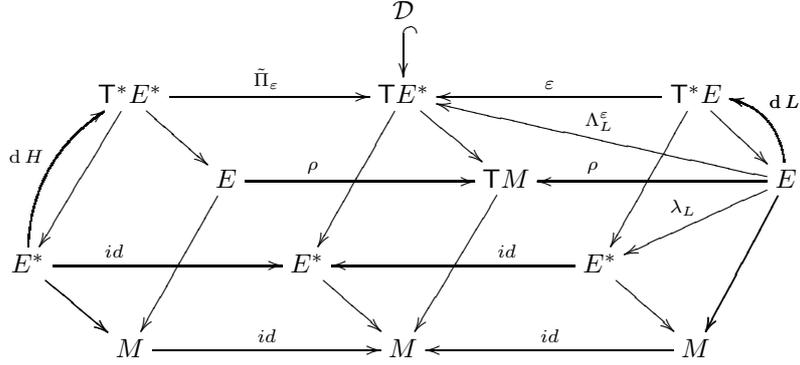
\xymatrix@C-10pt@R-5pt{
&&&& { \mathcal{D}}\ar@{ (->}[d]&&&&\\
 &\sT^\ast E^\ast \ar[rrr]^{\tilde\zP_\ze}
\ar[ddl] \ar[dr]
 &  &  & \sT E^\ast \ar[ddl] \ar[dr]
 &  &  & \sT^\ast E \ar[ddl] \ar[dr]
\ar[lll]_{\varepsilon}
 & \\
 & & E \ar[rrr]^/-20pt/{\zr}\ar[ddl]
 & & & \sT M\ar[ddl]
 & & & E\ar[lll]_/+20pt/{\zr}\ar[ddl]\ar[ddl]_{}\ar[ddl]_{}\ar@/_1pc/[ul]_{\xd L}\ar@/_1pc/[ul]_{\xd L}\ar[dll]_{\lambda_L} \ar[ullll]_{\Lambda_L^\ze}
 \\
 E^\ast\ar[rrr]^/-20pt/{id} \ar[dr]\ar[dr]^{}\ar@/^1pc/[uur]^{\xd H}
 & & & E^\ast\ar[dr]
 & & & E^\ast\ar[dr]\ar[lll]_/-20pt/{id}
 & & \\
 & M\ar[rrr]^{id}
 & & & M & & & M\ar[lll]_{id} &
}$$

\bigskip\noindent
The left-hand side is Hamiltonian with Hamiltonians being functions $H:E^*\to\R$,
the right-hand side is Lagrangian with Lagrangians being functions $L:E\to\R$, and the phase dynamics $\mathcal{D}$ lives in the middle and is understood as an implicit dynamics (\emph{Lagrange equation}) on the phase space $E^\ast$, i.e. a subset of $\sT E^\ast$.
Solutions of the Lagrange equations are `phase trajectories' of the system: $\zb:\R\ra E^\ast$ which satisfy $\st(\zb)(t)\in \mathcal{D}$. Here $\st$ is the canonical tangent prolongation of the curve $\zb$.

The dynamics $\mathcal{D}=\mathcal{D}_L$ generated by the Lagrangian $L$ is simply $\mathcal{D}=\Lambda^\ze_{L}(E)$, where $\Lambda^\ze_{L}: E\longrightarrow
\sT E^\ast$, $\zL^\ze_{L}=\ze\circ\xd L$ is the so called {\em Tulczyjew differential}.

Similarly, the Hamiltonian dynamics is $\mathcal{D}=\mathcal{D}_H=\tilde\zP_\ze(\xd H(E^\ast))$.
In this picture, looking for Hamiltonian description of the Lagrangian dynamics is looking for
a Hamiltonian $H$ such that $\mathcal{D}_H=\mathcal{D}_L$.
This formalisms includes the physically relevant relation between momenta and velocities given by
the \emph{Legendre map} $\lambda_L:E\longrightarrow E^\ast$, $\lambda_L=\tau_{E^\ast}\circ\ze\circ\xd L$, covered by the Tulczyjew differential. The Legendre map actually does not depend on the algebroid structure but only on the Lagrangian.

Note finally that the above formalisms can still be generalized to
include constraints (cf. \cite{Grabowska:2011}) and that a rigorous optimal
control theory on Lie algebroids can be developed as well
\cite{Grabowski:2011}.

\medskip
Since in some cases, instead of the full phase dynamics, people are satisfied with the \emph{Euler-Lagrange equations}, let us define them in our framework. They are first order  equation for curves {\ $\gamma: \R\rightarrow E$} in $E$:
$$(E_L):\qquad \st(\zl_L\circ\gamma)=\zL^\ze_L\circ\gamma.$$
The equation $(E_L)$ simply means that $\zL^\ze_L\circ\zg$ is an admissible curve in $\sT E^\ast$, thus it is the
tangent prolongation of $\zl_L\circ\zg$. In local coordinates, $\mathcal{D}$ has the parametrization by $(x^a,y^k)$
via $\zL^\ze_L$ in the form (cf. (\ref{F1.4}))
\be\zL^\ze_L(x^a,y^i)= \left(x^a,\frac{\partial L}{\partial y^i}(x,y),
\zr^b_k(x)y^k, c^k_{ij}(x) y^i\frac{\partial L}{\partial y^k}(x,y) + \zs^a_j(x)\frac{\partial L}{\partial
x^a}(x,y)\right) \label{F1.4a}\ee and the equation $(E_L)$, for $\zg(t)=(x^a(t),y^i(t))$, reads
\be (E_L):\qquad\frac{\xd x^a}{\xd t}=\zr^a_k(x)y^k, \quad
\frac{\xd}{\xd t}\left(\frac{\partial L}{\partial y^j}\right)= c^k_{ij}(x) y^i\frac{\partial L}{\partial y^k}
+ \zs^a_j(x)\frac{\partial L}{\partial x^a}\,.\label{EL2}\ee
Note that for the canonical Lie algebroid $E=\sT M$ and the adapted local coordinates $(x^a,y^i)=(x^a,\dot x^b)$ we have $\zr^a_b=\zs^a_b=\zd^a_b$ and $c^k_{ij}=0$, so that we recover the classical Euler-Lagrange equations
$$\frac{\xd}{\xd t}\left(\frac{\partial L}{\partial \dot x^i}\right)= \frac{\partial L}{\partial x^i}\,.
$$
The solutions are automatically \emph{admissible curves} in $E$, i.e. curves satisfying $\zr(\zg(t))=\st(\zt\circ\zg)(t)$.
As a curve in the canonical Lie algebroid $E=\sT M$ is admissible if and only if it is a tangent prolongation of
its projection on $M$, first order differential equations for admissible curves (paths) in $\sT M$ may be viewed as certain second order differential equations for curves (paths) in $M$. This explains why, classically, the Euler-Lagrange equations are regarded as second order equations.

\begin{remark}
In this understanding of first order mechanics following Tulczyjew, there is no need for ad-hoc constructions nor Lie algebroid prolongations, Poincar\'e-Cartan forms, Cartan sections, and symplectic algebroids. However, as we shall see, for higher order mechanics on a Lie algebroid there is a fundamental role for Lie algebroid prolongations via the linearisation functor, at least in the formalism we shall develop here. However, our use of Lie algebroid prolongations significantly differs from the r\^ole they play in some approaches to the first order mechanics on Lie algebroids.
\end{remark}

\subsection{Vakonomic constraints}\label{sebsec:affine constraints}
There is an extensive literature concerning vakonomic constraints in the Lagrangian and Hamiltonian formalism. We will skip presenting older concepts and use the elegant geometric approach, based on some ideas of Tulczyjew (see e.g. \cite{Tulczyjew:1977}), as described in \cite{Grabowska:2008}. It is ideologically much simpler than many others and works very well also for mechanics on algebroids. We will devote  a page to present main points of this approach.

\smallskip
Let us recall first that with any
submanifold $S$ in $E$ and any function $L:S\ra \R$ one can associate canonically a
lagrangian submanifold $S_L$ in $\sT^\ast E$ defined by
$$S_L=\{ \za_e\in\sT^\ast_eE: e\in S\text{\ and\ }\la\za_e,v_e\ran=\xd L(v_e)\text{\ for every\ }
v_e\in\sT_eS\}\,.$$ If $S=E$, then $S_L=\xd L(E)$, i.e. $S_L$ reduces to the image of
$\xd L$.
The vakonomically constrained phase dynamics is just $\mathcal{D}=\ze(S_L)\subset\sT E^*$.
With ${\cP L}:S\rRelation\sT^*E$  denote the relation
\be\label{lagrangerelation}
{\cP L}=\{ (e,\za)\in S\ti\sT^* E: \za\in S_L\ \&\ \zp_E(\za)=e\}\,.
\ee
Note that we use the notation ``$\rRelation$'' to reinforce the fact that we are dealing with relations rather than genuine maps. We will adopt this notation for relations throughout the remainder of this paper.\smallskip

We say that a curve $\zg:\R\ra S$ satisfies the \emph{vakonomic Euler-Lagrange equations}
associated with the Lagrangian $L:S\ra\R$ if and only if $\zg$ is $\ze\circ{\cP L}$-related to an admissible curve, i.e. to a curve which is the tangent prolongation of a curve in $E^\ast$ (c.f. \cite{Grabowska:2008}). In particular, the vakonomic Euler-Lagrange equations depend only on the the Lagrangian as a function on the constraint $S$ that differs vakonomic constraints from  nonholonomic of mechanical type ones. As admissible curves in $\sT E^*$ are exactly those whose tangent prolongations lie in the set $\sT^2 E^*\subset\sT\sT E^*$ of holonomic vectors, the vakonomic Euler-Lagrange dynamics on $E$ can be described ny means of the set $\wt{\zL^\ze_L}^{-1}(\sT^2E^*)$, where $\wt{\zL^\ze_L}=\ze\circ{\cP L}$ is the \emph{Tulczyjew differential} relation.

\begin{remark}
There is an alternative point of view \cite{Grabowska:2008} in which the vakonomic Euler-Lagrange equations are not
equations on curves in $S$ but rather on curves in $S_L$. Their projections to $S$ give curves satisfying vakonomic Euler-Lagrange equation in the previous sense. The advantage of this approach is that
the equations are `less implicit' as the `Lagrange multipliers' do not appear.
\end{remark}

It is easy to see that the lagrangian submanifold $S_L$ can be obtained also as the image of $\xd L(S)\subset\sT^*S$ under the symplectic relation $r_S:\sT^*S\rRelation\sT^*E$ associated with the embedding $\zi_S:S\hookrightarrow E$. Denoting the composition of this relation with $\ze$ as $\ze_S$, we get the phase dynamics $\cD$ in the form ${\cD}=\ze_S(\xd L(S))$ that gives the Lagrangian part of our formalism  completely analogous to the unconstrained one. We will sometimes drop the subscript $S$ if
$S$ is fixed and the meaning of $\zi=\zi_S:S\to E$ etc. is clear.

\medskip
In the case when $S=A$ is an affine subbundle of $E$ (assume for simplicity that $A$ is supported on the whole $M$), the relation $\ze_A$ covers the map $\sv(\zi)^*:E^*\to \sv^*(A)$ which is dual to the map $\sv(\zi):\sv(A)\to E$ associated with the affine embedding $\zi:A\to E$. With respect to the second fibration, it covers the restriction $\zr_A:A\to\sT M$ of the anchor map $\zr:E\to\sT M$. We can consider the Hamiltonian part just using Theorem \ref{th:1}. In this way we get the \emph{reduced Tulczyjew triple} for an affine vakonomic constraint:
\be\label{reduced-vaconomic}
\begin{array}[c]{c}
\xymatrix@C-10pt@R-5pt{
&&&& { \mathcal{D}}\ar@{ (->}[d]&&&&\\
 &\sP(A^\dag) \ar@{-|>}[rrr]^{\ze_A\circ\cR_A^{-1}}
\ar[ddl] \ar[dr]
 &  &  & \sT E^\ast \ar[ddl] \ar[dr]
 &  &  & \sT^\ast A \ar[ddl] \ar[dr]
\ar@{-|>}[lll]_{\varepsilon_A}
 & \\
 & & A \ar[rrr]^/-20pt/{\zr_A}\ar[ddl]
 & & & \sT M\ar[ddl]
 & & & A\ar[lll]_/+20pt/{\zr_A}\ar[ddl]\ar[ddl]_{}\ar[ddl]_{}\ar@/_1pc/[ul]_{\xd L}
 \\
 \sv^\ast(A)\ar[dr]\ar[dr]^{}\ar@/^1pc/[uur]^{\uxd H}
 & & & E^\ast\ar[lll]_/+20pt/{\sv(\zi)^*}\ar[rrr]^/+20pt/{\sv(\zi)^*}\ar[dr]
 & & & \sv^\ast(A)\ar[dr]
 & & \\
 & M\ar@{=}[rrr]
 & & & M & & & M\ar@{=}[lll] &
}\end{array}
\ee

\medskip
Here, the Hamiltonian is a section of the AV-bundle $\zz: A^\dag\longrightarrow \sv(A)^\ast$. It is easy to see that in the case when the Lagrangian is hyperregular, i.e. the vertical derivative $\xd^vL$ viewed as the Legendre map
$\zl_L:A\to\sv^*(A)$ is a diffeomorphism, the lagrangian submanifold $\xd L(A)\subset\sT^*A\simeq\sP(A^\dag)$ has the
Hamiltonian description, $\xd L(A)=\uxd H(\sv^*(A))$, with the Hamiltonian section $H:\sv^*(A)\to A^\dag$ such that $H(v^*_x)$ is the unique affine function on $A_x$ with the linear part $v^*_x$ and satisfying
\be\label{Ham}
H(v^*_x)(\zl^{-1}_L(v^*_x))=L(\zl^{-1}_L(v^*_x))\,.
\ee
If the affine bundle $A$ is linear, $A^\dag=A^*\ti\R$ is trivial, so the Hamiltonian (\ref{Ham})
understood as a genuine function on $A^*$ looks like the completely classic one (c.f. Remark \ref{r1}):
$$H(v^*_x)=\la v^*_x,\zl_L^{-1}(v^*_x)\ran-L(\zl_L^{-1}(v^*_x))\,.
$$

\subsection{Tulczyjew triples for higher order mechanics}\label{sec:triples for mechanics}

Consider $k$th order mechanics, where in the ``unreduced approach'' the Lagrangian $L$  is understood as a function on the submanifold $\sT^k Q\subset\sT\sT^{k-1} Q$. The unreduced Tulczyjew triple in this case is
\[
\begin{array}[c]{c}
 \xymatrix@C-30pt@R-5pt{
& \sT^\ast\sT^\ast\sT^{k-1} Q\ar[rrr]^{\tilde\zP_{\ze}}\ar[dl]\ar[ddr]|!{[drr];[dl]}\hole & & & \sT\sT^\ast\sT^{k-1} Q \ar[dl]\ar[ddr]|!{[drr];[dl]}\hole
& & & \sT^\ast\sT\sT^{k-1}  Q\ar[lll]_{\ze} \ar[dl]\ar[ddr]& \\
\sT^\ast\sT^{k-1} Q\ar[ddr] & & & \sT^\ast\sT^{k-1} Q\ar@{=}[lll]\ar@{=}[rrr]\ar[ddr] & & & \sT^\ast\sT^{k-1} Q \ar[ddr]& & \\
& & \sT\sT^{k-1} Q\ar[dl] & & & \sT\sT^{k-1} Q\ar@{=}[lll]|!{[ull];[dl]}\hole  \ar@{=}[rrr]|!{[ur];[drr]}\hole \ar[dl] & & & \sT\sT^{k-1} Q\ar[dl]  \\
& \sT^{k-1} Q & & &\sT^{k-1} Q\ar@{=}[lll]\ar@{=}[rrr] & & & \sT^{k-1} Q &
}
\end{array},
\]
where $\ze$ determines the canonical Lie algebroid structure on $\sT\sT^{k-1}Q$ and $\zP_\ze$ is the Poisson tensor corresponding to the canonical symplectic form $\omega_{\sT^{k-1} Q}$ on $\sT^\ast\sT^{k-1} Q$.

Starting form local coordinates $q=(q^a)$ in $Q$ and $(v)$ in $\sT^{k-1} Q$, where
 $v=(q,\dot q,\ddot q,\dots, \overset{(k-1)}{q})$, we get natural coordinates
$$\begin{array}{ll}\vspace{5pt}
( v,  \zd v) & \text{in }\sT\sT^{k-1} Q ,\\ \vspace{5pt}
(v, p) & \text{in }\sT^\ast\sT^{k-1} Q,\\ \vspace{5pt}
( v, \zd v, \zp, \zd\zp) & \text{in }\sT^\ast\sT\sT^{k-1} Q, \\ \vspace{5pt}
(v, p, \dot v, \dot p) & \text{in }\sT\sT^\ast\sT^{k-1} Q, \\ \vspace{5pt}
(v, p, \varphi,  \psi) & \text{in }\sT^\ast\sT^\ast\sT^{k-1} Q\,,
\end{array}$$
in which $\ze_{\sT^{k-1} Q}$ is just the identification of coordinates $\zd v=\dot v, \zp=\dot p, \zd\zp=p$, and
$\zP_{\sT^{k-1} Q}$ corresponds to the identification $\varphi=\dot v,
\psi=-\dot p$. It will be convenient to write $\overset{(i)}{q}$ as $v_i$, so that the full coordinate system on $\sT^{k-1} Q$ can be written as $(v_i^a)$.

The submanifold $S=\sT^k Q$ in $\sT\sT^{k-1} Q$ is given by the condition ${\zd v}_{i-1}=v_i$, $i=1,\dots,k-1$, with the embedding $\hat\zr_k:\sT^kQ\to\sT\sT^{k-1}Q$, which in this case coincides with the anchor map,
$${\hat\zr_k}\left(q,\dot q,\dots, \overset{(k)}{q}\right)=\left(q,\dot q,\dots, \overset{(k-1)}{q},\dot q,\ddot q,\dots,\overset{(k)}{q}\right)=(v,\zd v)\,.$$
We view $S$ as an affine vakonomic constraint, so the reduced triple is
\be\label{hreduced-classical}
\begin{array}[c]{c}
\xymatrix@C-30pt@R-5pt{
& \sP(\sT^kQ)^\dag\ar@{-|>}[rrr]\ar[dl]\ar[ddr] & & & \sT\sT^\ast\sT^{k-1} Q \ar[dl]\ar[ddr]
& & & \sT^\ast\sT^k Q \ar@{-|>}[lll]\ar[dl]\ar[ddr]& \\
\sT^{k-1} Q\times_Q\sT^\ast Q\ar[ddr]+<-4.5ex,2ex> & & & \sT^\ast\sT^{k-1} Q\ar[lll]\ar[rrr]\ar[ddr] & & & \sT^{k-1} Q\times_Q\sT^\ast Q \ar[ddr]+<-4.5ex,2ex>& & \\
& & \sT^k Q\ar[dl]\ar@{ (->}[rrr] & & & \sT\sT^{k-1} Q\ar[dl] 
& & & \sT^k Q\ar[dl]\ar@{ (->}[lll]  \\
&\sT^{k-1} Q \  & & &\sT^{k-1} Q\ \ar@{=}[lll]\ar@{=}[rrr] & & &\ \sT^{k-1} Q &
}
\end{array}.
\ee
Hence, according to the general scheme, the Lagrangian function $L=L(q,\dots,\overset{(k)}{q})$ generates the following (lagrangian) submanifold in $\sT^\ast\sT\sT^{k-1} Q$:
$$(\sT^k Q)_L=\left\{\,(v, \zd v, \pi, \zd \pi):\quad
{\zd v}_{i-1}=v_i,\;\; {\zp}_i+\zd{\zp}_{i-1}=\frac{\partial L}{\partial \overset{(i)}{q}}\,, i=1,\dots, k-1\,,{\zp}_0=\frac{\partial L}{\partial {q}}\,, \zd{\zp}_{k-1}=\frac{\partial L}{\partial \overset{(k)}{q}}\right\}\,.$$
Here, the conditions ${\zd v}_{i-1}=v_i$, $i=1,\dots,k-1$, mean that $(v,\zd v)$ is in the image of $\hat\zr_k$, and partial derivatives of the Lagrangian are taken in the point $\hat\zr_k^{-1}(v,\zd v)$.
The phase dynamics $\mathcal{D}=\ze((\sT^{k-1} Q)_L)$ is then
$$\mathcal{D}=\left\{\, (v, p, \dot v, \dot p):\quad
\dot v_{i-1}=v_i,\;\; \dot p_i+p_{i-1}=\frac{\partial L}{\partial\overset{(i)}{q}},i=1,\dots, k-1\,,{\dot p}_0=\frac{\partial L}{\partial {q}}\,,p_{k-1}=\frac{\partial L}{\partial \overset{(k)}{q}}\right\}.$$
We understand $\mathcal{D}$ as an implicit first order differential equation for curves in $\sT^\ast\sT^{k-1} Q$. A curve $\zg(t)=(q(t),\dot q(t),\dots,\overset{(k)}{q}(t))$ is $\wt{\zL^\ze_L}$ related with a curve $\zb(t)=(v(t), p(t), \dot v(t), \dot p(t))$ in $\cD$ if and only if
$$v(t)=\left(q(t),\dot q(t),\dots,\overset{(k-1)}{q}(t)\right)\,,\quad\dot v(t)=\left(\dot q(t),\ddot q(t),\dots,\overset{(k)}{q}(t)\right)\,,
$$
and
$$\dot p_i(t)+p_{i-1}(t)=\frac{\partial L}{\partial\overset{(i)}{q}}(\zg(t))\ \text{for}\ i=1,\dots, k-1\,,\quad {\dot p}_0(t)=\frac{\partial L}{\partial {q}}(\zg(t))\,,\quad p_{k-1}(t)=\frac{\partial L}{\partial \overset{(k)}{q}}(\zg(t)).
$$
Assuming additionally that $\zb$ is admissible, we get  equations
\bea\label{hEL}
&\frac{\xd}{\xd t}\overset{(i)}{q}=\overset{(i+1)}{q}\,,\ i=0,\dots,k-1;\\
&p_{k-1}=\frac{\partial L}{\partial \overset{(k)}{q}}(q,\dots,\overset{(k)}{q})\,,
\quad p_{i-1}=\frac{\partial L}{\partial \overset{(i)}{q}}(q,\dots,\overset{(k)}{q})-
\frac{\xd}{\xd t}(p_i)\ \text{for}\ i=1,\dots,k-1;\\
&\frac{\xd}{\xd t}(p_0)=\frac{\partial L}{\partial {q}}(q,\dots,\overset{(k)}{q})\,,
\eea
that can be rewritten as the Euler-Lagrange equations in the traditional form:
\bea\label{higherEL}
&\overset{(i)}{q}=\frac{\xd^i}{\xd t^i}q\,,\ i=1,\dots,k\,,\\
&\frac{\partial L}{\partial {q}}(q,\dots,\overset{(k)}{q})-\frac{\xd}{\xd t}\left(\frac{\partial L}{\partial {\dot q}}(q,\dots,\overset{(k)}{q})\right)+\cdots+(-1)^k\frac{\xd^k}{\xd t^k}\left(\frac{\partial L}{\partial \overset{(k)}{q}}(q,\dots,\overset{(k)}{q})\right)=0\,.
\eea
These equations can be viewed as a system of differential equations of order $k$ on $\sT^kQ$ or, which is the standard point of view, as ordinary differential equation of order $2k$ on $Q$.


\section{Mechanics on graded bundles}\label{sec:MechanicsGradedBundles}
The basic idea is to follow the construction of the Tulczyjew triple for the classical higher order mechanics in which we replace $\sT^kM$ with an arbitrary graded bundle $F_k$, and the canonical embedding $\sT^kM\hookrightarrow\sT\sT^{k-1}M$ with the canonical embedding $F_k\hookrightarrow D(F_k)$. Here the linearisation $D(F_k)$ is equipped with a weighted algebroid structure $\ze:\sT^*D(F_k)\to\sT D(F_k)$. For simplicity and with examples in mind we will work with skew/Lie weighted algebroids but general weighted algebroids can be used as well.\smallskip

We therefore consider mechanics on the algebroid $D(F_k)$ with the affine constraint $F_k\subset D(F_k)$. In consequence, we understand the phase space of a mechanical system on $F_{k}$ to be $D^{*}(F_{k})$  and the (implicit) phase equations as subsets of $\sT D^{*}(F_{k})$.

\smallskip
The corresponding reduced Tulczyjew triple mimics that of (\ref{reduced-vaconomic}), where the affine bundle $A$ is $\zt^k:F_k\to F_{k-1}$, and the vector bundle $E\to M$ is replaced with $D(F_k)\to F_{k-1}$. Note that in this case $\sv^*(A)$ is the Mironian $\textnormal{Mi}(F_{k}) \simeq F_{k-1} \times_{M} \bar{F}_{k}^{*}$. Denoting the affine dual of $\zt^k:F_k\to F_{k-1}$ simply with $F_k^\dag$, we get the triple as follows.

\be\label{graded_Tulczyjew}
\begin{array}[c]{c}
\xymatrix@C-20pt@R-5pt{
&&&& { \mathcal{D}}\ar@{ (->}[d]&&&&\\
 &\sP(F_k^\dag) \ar@{-|>}[rrr]^{\wh{\zP}_{\hat\ze}}
\ar[ddl] \ar[dr]
 &  &  & \sT D^\ast(F_k) \ar[ddl] \ar[dr]
 &  &  & \sT^\ast F_k \ar[ddl] \ar[dr]
\ar@{-|>}[lll]_{\hat\varepsilon}
 & \\
 & & F_k \ar[rrr]^/-20pt/{\hat\zr}\ar[ddl]
 & & & \sT F_{k-1}\ar[ddl]
 & & & F_k\ar[lll]_/+20pt/{\hat\zr}\ar[ddl]\ar[ddl]_{}\ar[ddl]_{}\ar@/_1pc/[ul]_{\xd L}
 \ar[dll]_{\lambda_L} \ar@{-|>}[ullll]_{\Lambda_L^{\hat\ze}}
 \\
 \Mi(F_k)\ar[dr]\ar[dr]^{}\ar@/^1pc/[uur]^{\uxd H}
 & & & D^\ast(F_k)\ar[lll]\ar[rrr]\ar[dr]
 & & & \Mi(F_k)\ar[dr]
 & & \\
 & F_{k-1}\ar@{=}[rrr]
 & & & F_{k-1} & & & F_{k-1}\ar@{=}[lll] &
}\end{array}
\ee

\medskip\noindent
Here, we write $\hat\ze$ instead of $\ze_{F_k}$ and $\hat\zr$ instead of $\zr_{F_k}$. The relation $\wh{\zP}_{\hat\ze}$ can also be written as $\hat\ze\circ\cR_k$, where $\cR_k$ is the canonical isomorphism identifying the vector-affine bundles $\sT^*F_k$ and $\sP(F_k^\dag)$. Now, the generation of the phase dynamics out of a Lagrangian function or a Hamiltonian section and the construction of the Euler-Lagrange equations is subject of the general scheme for affine vakonomic constraints in algebroids, as described in subsection \ref{sebsec:affine constraints}. However, the fact that the algebroid is a weighted algebroid of the graded bundle $F_k$ puts additional flavour to the scheme, making the whole picture similar to the classical case $F_k=\sT^kM$.

\subsection{The Lagrangian formalism on graded bundles}\label{subsec:lagrangian graded bundles}

Let us employ natural homogeneous coordinates $\displaystyle\left( \bar{X}^{\za}_W; \bar{P}_{\zb}^{k+1-W}\right)$ on $\sT^{*}F_{k}$, so that $\bar{X}^{\mu}_{u}$ serve as local coordinates on $F_{k-1}$.  The reason for the ``bar'' in the notation will become clear as we will be dealing with relations. In these local coordinates  the section $dL: F_{k} \rightarrow \sT^{*}F_{k}$  is given by:
$$ \bar{P}_{\za}^{k+1-W} \circ dL = \frac{\partial L}{\partial \bar{X}_{W}^{\za}}\,.
$$

The canonical inclusion $\iota: F_{k} \longhookrightarrow  D(F_{k})$ induces a graded symplectic relation $r: \sT^{*} F_{k}\rRelation \sT^{*} D(F_{k})$. This, in turn, produces the graded relation $\mathcal{P}L := r \circ dL :F_{k} \rRelation \sT^{*} D(F_{k})$.   By employing homogeneous local coordinates $\displaystyle (X^\za_u,Y^I_U,P^{k+1-u}_\zb,\zP^{k+1-U}_J)$, this relation is given by:

\bea\label{lag}X^{\mu}_u &=& \bar{X}^{\mu}_u\,,\quad Y^{I}_{U} = U \: \bar{X}^{I}_{U}\,,\\
 P_{\za}^{W+1} &=& \frac{\partial L}{\partial \bar{X}_{k-W}^{\za}}(\bar X)- (k-W)\: \zP_{\za}^{W+1}\,,
\eea
where we employ the convention that coordinates with degrees outside the range are zero: $P_\zb^1=0$ and $\zP_\zb^{k+1}=0$.

\medskip
 The Lagrangian produces the  \emph{phase dynamics}  $\mathcal{D}_L$ understood as the image of  $F_{k}$ under the relation $\zL^\ze_{L} = \epsilon \circ \mathcal{P}L$:
\begin{equation}
\mathcal{D}_L =  \zL^\ze_{L}(F_{k}) \subset \sT  D^{*}(F_{k}).
\end{equation}
\noindent  The relation $\zL^\ze_{L}$ we will refer to as the \emph{weighted Tulczyjew differential of} $L$. Note that the phase dynamics explicitly depends on the weighted algebroid structure carried by $F_{k}$. In the preceding we will assume that this weighted algebroid structure is skew/Lie to slightly simplify the local expressions, though this specialisation is not fundamentally required in the formalism.  In particular, note  we do not  need the Jacobi identity for the associated bracket structure on sections and so weighted skew algebroids will more that  suffice for this formulation of higher order geometric mechanics.   \smallskip

By using  natural homogeneous local coordinates $\displaystyle (X^{\mu}_u,\zP^{U}_I,\zd X^{\nu}_{u+1},\zd\zP^{U+1}_J)$ in $\sT D^*(F_k)$, the relation $\Lambda_{L}^{\epsilon}: F_{k}  \rRelation \sT D^{*}(F_{k}) $ be described by:
\bea\label{eqn:lagrangian relation}
X^\mu_u&=&\bar X^\mu_u\,,\\
\nonumber \zd X^\nu_{U}&=&(U-u)\zr[u]^\nu_I(\bar X)\bar X^I_{U-u}\,,\\
\nonumber k\zP_i^1&=& \frac{\partial L}{\partial \bar{X}_{k}^{i}}(\bar X) \\
\nonumber \zd \zP^{U+1}_J&=&\zr[u]^{\mu}_J(\bar X)\left(\frac{\partial L}{\partial \bar{X}_{k-U+u}^{\mu}}(\bar X)- (k-U+u)\: \zP_{J}^{U-u+1}\right)\\
&&+U'C[u]^K_{IJ}(\bar X)\bar X^I_{U'}\zP^{U+1-U'-u}_K. \nn
\eea

 We understand the phase dynamics to be the first order implicit differential equation  $\mathcal{D}_L=\zL_L^\ze(F_k) \subset \sT D^{*}(F_{k})$ on the phase space $D^{*}(F_{k})$. A curve $\displaystyle\zb(t)= (X^{\mu}_u(t),\zP^{U}_I(t))\in D^*(F_k)$
 is a solution of this equation if and only if its tangent prolongation
 \be\label{prolongation}(X^\mu_u(t),\zP^{U}_I(t),\dot X^\nu_u(t),\dot\zP^{U}_I(t))\in \sT D^*(F_k)
 \ee
lies in $\mathcal{D}_L$. Here, dots have the meaning of genuine time-derivatives. Therefore, a curve $\zg(t)=( \bar{X}^{\za}_W(t))$ in $F_k$ satisfies the Euler-Lagrange equation if and only if it is $\zL_L^\ze$-related with an admissible curve (\ref{prolongation}), i.e.
\bea\label{admissibility}
 \dot{X}^\nu_{U}&=&(U-u)\zr[u]^\nu_I(\bar X)\bar X^I_{U-u}\,,\\
 k\zP_i^1&=& \frac{\partial L}{\partial \bar{X}_{k}^{i}}(\bar X) \\
 \dot{\zP}^{U+1}_J&=&\zr[u]^{\mu}_J(\bar X)\left(\frac{\partial L}{\partial \bar{X}_{k-U+u}^{\mu}}(\bar X)- (k-U+u)\: \zP_{J}^{U-u+1}\right) \label{zP^k}\\
&&+U'C[u]^K_{IJ}(\bar X)\bar X^I_{U'}\zP^{U+1-U'-u}_K. \nn
\eea
The first equation means that the curve $\gamma: \mathbb{R}\to F_{k}\hookrightarrow D(F_k)$ is  \emph{admissible}, i.e.
\begin{equation*}
\hat{\rho} \circ \gamma = \st {\zg_{k-1}}\,,
\end{equation*}
\noindent where $\hat{\rho} := \rho \circ \iota : F_{k} \rightarrow \sT F_{k-1}$ is the anchor map, ${\gamma}_{k-1} = \tau^{k} \circ \gamma$ is the curve on $F_{k-1}$ underlying $\gamma$, and $\st  {\zg_{k-1}}$ is the tangent prolongation of the said  underlying curve.

The rest of equations defines an implicit differential equation for curves on $F_k$, that is standard for vakonomic equations, on additional parameters $\zP^W_I$. These parameters are fixed if we understand Euler-Lagrange equations as equations on the lagrangian submanifold $(F_k)_L\subset\sT^*D(F_k)$. The latter equations are, in `good' cases, of order $k$ (c.f. (\ref{higherEL})). Indeed, if the matrix $(\zr[0]^\mu_I(\bar X))$ is invertible, we can express
each $\zP_{I}^{U+1}$, $U=1,\dots, k-1$, as a function $\zP_{I}^{U+1}=F_I^U(\dot \zP^{U},\zP^{U},\zP^{U-1},\dots,\zP^1,\bar X)$ of $\dot \zP_{I}^{U}$ and of  variables $\zP^{U'}$ of lower weight and $\bar X$.
As
$$k\zP_i^1= \frac{\partial L}{\partial \bar{X}_{k}^{i}}(\bar X)\,,$$
we get inductively that $\zP_{\za}^{U+1}$ is a function of
$$\frac{\xd^U}{\xd t^U}\left(\frac{\partial L}{\partial \bar{X}_{k}^{a}}(\bar X)\right)$$
and derivatives of $\bar X$ of order $< U$, $U=1,\dots, k-1$.
Since, according to (\ref{zP^k}), $\dot\zP^{k}_\zb$ is a function of variables $\zP$ and $\bar X$,
we get an equation on derivatives of $\bar X$ of order $\le k$. In this situation a curve on $F_k$
has at most one `prolongation' to a corresponding curve in the lagrangian submanifold $(F_k)_L\subset \sT^*D(F_k)$, so  the concepts of Euler-Lagrange equations understood as equations for curves on $F_k$ or $(F_k)_L$ coincide. This is the case of the standard higher order Lagrangian mechanics on $\sT^kM\subset\sT\sT^{k-1}M$.\smallskip

Observe additionally that the admissibility equation (\ref{admissibility}) puts additional relations on variables $\bar X$ which are differential equations of order $k$ in `good' cases. Indeed, we get inductively  the $k$th derivative of $\bar X_0$ as a function of lower order derivatives of $\bar X$. Geometrically it means that the series of anchors $F_{i}\to\sT F_{i-1}$
gives rise to a map $\hat\zr^k:F_k\to \sT^kF_0$ and the $k$th prolongation of the curve $\zg_0=\zt^k_0\circ\zg$ being the projection of the admissible curve $\zg$ on $F_k$ equals $\hat\zr^k\circ\zg$:
\be\label{hadm}
\hat\zr^k\circ\zg=\st^k {\zg_{0}}\,.
\ee
In the canonical case $\sT^kM\subset\sT\sT^{k-1}M$, admissible curves in $\sT^lM$ are just $k$th order prolongations of curves on $M$, so we get equations of order $2k$ on $M$.

\begin{example}\label{e1}
Let $\g$ be a real finite-dimensional Lie algebra with the structure constants $c_{ij}^k$ relative to a chosen basis $e_i$, and put $F_2=\g_2=\g[1]\ti\g[2]$ (c.f. Example \ref{e-lalg}). Note that this graded bundle is actually a \emph{graded space} and its linearisation, carrying a canonical Lie algebroid structure, is an example of a \emph{weighted Lie algebra} in the terminology of (see \cite{Bruce:2014}), as there are no coordinates of weight 0. The basis induces coordinates $(x^i,z^j)$ on $\g_2$ and coordinates
$(x^i,y^j,z^k)$ on $D(\g_2)=\g[1]\ti\sT\g[1]=\g[1]\ti\g[1]\ti\g[2]$ for which the embedding $\zi:\g_2\hookrightarrow D(\g_2)$ takes the form $\zi(x,z)=(x,x,z)$ and the vector bundle projection is $\zt(x,y,z)=x$.
The Lie algebroid structure on $D(\g_2)$ is the product of the Lie algebra structure on $\g$ and the tangent bundle $\sT\g$; the map $\ze:\sT^*D(\g_2)\to\sT D^*(\g_2)$ takes the form
\be
(x,y,z,\za,\zb,\zg)\mapsto(x,\zb,\zg,z,\ad_y^*\zb,\za)\,,
\ee
where $(\ad_y^*\zb)_j=c^k_{ij}y^i\zb_k$ \cite{Bruce:2014}.
Given a Lagrangian $L:\g_2\to\R$, the Tulczyjew differential relation $\zL_L^\ze:\g_2\to\sT D^*(\g_2)$ is
$$\zL_L^\ze(x,z)=\left\{\left(x,\zb,\frac{\pa L}{\pa z}(x,z),z,\ad_x^*\zb,\za\right): \za+\zb=\frac{\pa L}{\pa x}(x,z)\right\}\,.
$$
A curve in $\mathcal{D}_L=\zL_L^\ze(\g_2)$ is admissible if and only if (dots are now time derivatives)
\beas
\dot x&=&z\,,\\
\dot \zb&=&\ad_x^*\zb\,,\\
\za&=&\frac{\xd}{\xd t}\left(\frac{\pa L}{\pa z}(x,z)\right)\,.
\eeas
Hence,
$$\zb=\frac{\pa L}{\pa x}(x,z)-\frac{\xd}{\xd t}\left(\frac{\pa L}{\pa z}(x,z)\right)$$
that leads to the Euler-Lagrange equations on $\g_2$:
\beas
\dot x&=&z\,,\\
\frac{\xd}{\xd t}\left(\frac{\pa L}{\pa x}(x,z)-\frac{\xd}{\xd t}\left(\frac{\pa L}{\pa z}(x,z)\right)\right)&=&\ad_x^*\left(\frac{\pa L}{\pa x}(x,z)-\frac{\xd}{\xd t}\left(\frac{\pa L}{\pa z}(x,z)\right)\right)\,.
\eeas
These equations are second order and induce the Euler-Lagrange equations on $\g$ which are of order 3:
$$\frac{\xd}{\xd t}\left(\frac{\pa L}{\pa x}(x,\dot x)-\frac{\xd}{\xd t}\left(\frac{\pa L}{\pa z}(x,\dot x)\right)\right)=\ad_x^*\left(\frac{\pa L}{\pa x}(x,\dot x)-\frac{\xd}{\xd t}\left(\frac{\pa L}{\pa z}(x,\dot x)\right)\right)\,.
$$
For instance, consider the `free' Lagrangian $L(x,z)=\half\sum_iI_i(z^i)^2$ induces on $\g_2$ the equations (we do not use the summation convention loosing control on indices)
\beas
\dot x&=&z\,,\\
I_j\ddot  z^j&=&\sum_{i,k}c^k_{ij}I_kx^iz^k\,,
\eeas
which are equivalent to the equations
$$I_j\dddot x^j=\sum_{i,k}c^k_{ij}I_kx^i\dot x^k$$
on $\g$.
The latter can be viewed as `higher Euler equations' (for a rigid body if $\g=\so$).
\end{example}

\subsection{The Hamiltonian formalism on graded bundles}\label{sec:Hamiltonian graded bundles}

 For the system associated with a Hamiltonian section $H:\Mi(F_k)\rightarrow  F_k^\dag$ we obtain the phase dynamics $\cD_H$ understood as the image of $\Mi(F_k)$ under the relation $\Lambda_H^{\hat\ze}=\tilde\zP_{\hat\ze}\circ\cP H$. The map
$\tilde\zP_{\hat\ze}:\sT^\ast D^\ast(F_k)\rightarrow \sT D^\ast(F_k)$ encodes the structure of weighted algebroid on $F_k$ while
$\cP H$ is the composition of $\uxd H: \Mi(F_k)\rightarrow \sP F_k^\dag$ with symplectic relation $s:\sP F_k^\dag\rRelation \sT^\ast D^\ast(F_k)$. The relation $s$ is the Hamiltonian counterpart of the relation $r:\sT^\ast F_k\rRelation\sT^\ast D(F_k)$ that we used in Lagrangian mechanics and which is just the phase lift of the inclusion $F_k\hookrightarrow D(F_k)$. The relation $s$ can be obtained as the composition of $r$ with two isomorphisms $\cR_k:\sT^\ast F_k\longrightarrow \sP F_k^\dag$ and
$\cR_{D(F_k)}: \sT^\ast D(F_k)\longrightarrow \sT^\ast D^\ast(F_k)$, more precisely
$$s=\cR_{D(F_k)}\circ r\circ\cR_k^{-1}.$$
The relation $s$ can be however obtained independently of $r$. First notice that the inclusion $F_k\hookrightarrow D(F_k)$ is affine,
i.e. the image of a fibre of $\zt^k:F_k\to F_{k-1}$ is an affine subspace of the appropriate fibre of the vector bundle $D(F_k)\rightarrow F_{k-1}$. It means that the vector hull $\hat F_k$ is a vector subbundle of $D(F_k)$. The vector dual of $\hat F_k\hookrightarrow D(F_k)$ is the projection $D^\ast(F_k)\rightarrow F_k^\dag.$ The phase lift of that projection is one of the components of the relation $s$. To get $s$ we have to compose it with the symplectic reduction $\sT^\ast F_k^\dag\rRelation \sP F_k^\dag$. Of course, one should check if the phase lift of projection and reduction are composable. This can be done in coordinates.

For the coordinate expression of $s$ we will use coordinates:
\begin{align*}
& (X^\mu_u, \Theta^1_i, \zG^{k+1}_A, \zG^{k+1-w}_a, X^i_k) &\text{on } &\sP F_k^\dag, & \\
& (X^\mu_u, \zP^{k+1-w}_a, \zP^1_i, P^{k+1}_A, P^{k+1-w}_a, Y^a_w, Y^i_k) &\text{on } &\sT^\ast D^\ast(F_k).&
\end{align*}
The relation $s$ is described by the conditions
$$\Theta^1_i=k\zP^1_i,\quad P^{k+1}_A=\zG^{k+1}_A,\quad Y_w^a=wX^a_w,\quad Y^i_k=kX^i_k,\quad P^{k+1-w}_a=\zG^{k+1-w}_a+w\zP^{k+1-w}_a,$$
so $\cP H:\Mi(F_k)\rRelation \sT^\ast D^\ast(F_k)$ reads
$$\zP^1_i=\frac{1}{k} \Theta^1_i, \quad P^{k+1}_A=\frac{\partial H}{\partial X^A_0}, \qquad P^{k+1-w}_a=\frac{\partial H}{\partial X^{a}_w}
+w\zP^{k+1-w}_a,\quad Y^a_w=wX^a_w,\quad Y^i_k=k\frac{\partial H}{\partial \Theta^1_i}.$$
Finally, the coordinate expression for $\Lambda_H^{\hat\ze}$ is
\begin{align}
\delta X^\mu_U\circ \tilde\Lambda_\ze& =\rho[U-k]^\mu_i(X)k\frac{\partial H}{\partial \Theta^1_i} + \rho[U-w]^\mu_{a}(X) X^a_w \\
\delta \zP^{U+1}_J\circ\tilde\Lambda_\ze& =-\rho[U-k]^A_J\frac{\partial H}{\partial X^A_0}-\rho[U-k+w]^a_J\left(\frac{\partial H}{\partial X^a_w}+w\zP^{k+1-w}_a\right)\\
& \quad+ C[u]^K_{iJ}(X)k\frac{\partial H}{\partial \Theta^1_i}\zP^{U+1-k-u}+C[u]^K_{a J}(X)wX^a_w\zP^{U+1-w-u}.\notag
\end{align}

Usually we are given a system with Lagrangian function $L$ defined.  The question wether there exists the corresponding Hamiltonian section such that ${\cD}_H={\cD}_L$ arises naturally. The answer is very much the same as in the classical case for first order mechanics. It exists if Lagrangian is hyperregular i.e. the Legendre map $\lambda_L: F_k\rightarrow \Mi(F_k)$, $\lambda_L=\sT^\ast\zt^k\circ\xd L$, is a diffeomorphism.
In case it is not, there exists a generating family of sections parameterized by elements of $F_k$. Using the correspondence between sections of $F_k^\dag\rightarrow \Mi(F_k)$ and functions on $F_k^\dag$ we can write the generating family of functions as
\be h:F_k^\dag\times_{F_{k-1}} F_k\ni (\varphi, f)\longmapsto \varphi(f)-L(f)\in\R.\ee

\begin{example}\label{e-ham}
Let us consider the Hamiltonian side of the Lagrangian mechanics presented in Example \ref{e1}, i.e. for $F_2=\g_2=\{(x,z)\}\subset \g\ti\sT\g=\{(x,y,z)\}$ and the `free' Lagrangian $L(x,z)=\half\sum_iI_i(z^i)^2$ on $\g_2$. Since in this case the bundle $\zt_2:F_2=\g_2\to F_1=\g[1]$ is canonically linear with fibers $\g[2]$, the AV-bundle $\zz:\g_2^\dag\to\sv^*(\g_2)=\g[1]\ti\g^*[1]$ is trivial, $\g_2^\dag=\g[1]\ti\g^*[1]\ti\R$, and Hamiltonians can be understood as genuine functions on $\g[1]\ti\g^*[1]$. Note that, due to our convention, a function $H:\g[1]\ti\g^*[1]\to\R$ corresponds to the section $(x,\zvy,-H(x,\zvy))$ of $\g_2^\dag$. Our Lagrangian is hyperregular,
$$\zl_L:\g_2\to\sv^*(\g_2)\,,\ \zl_L(x,z)=(x,I\dot z)\,,
$$
where $(I\dot z)_i=I_iz^i$ (no summation convention here) is a diffeomorphism, so, according to (\ref{Ham}),
$$H(x,\zvy)=\la \zvy,\frac{1}{I}\dot\zvy\ran - L(\zl_L^{-1}(x,\zvy))=\sum_i\frac{\zvy_i^2}{I_i}
-\half\sum_iI_i\left(\frac{\zvy_i}{I_i}\right)^2=\half\sum_i\frac{\zvy_i^2}{I_i}\,.
$$
\end{example}


\section{Reductions}\label{sec:Higher Lagrangian Algebroid}
As remarked on in the introduction, higher order mechanics on Lie groups and Lie groupoids has received very little attention in the literature. The need to understand  higher order mechanics  on a Lie algebroid  naturally appears  in the context of reductions of higher order theories  on Lie groupoids with
Lagrangians that are invariant with respect to the groupoid multiplication.  However, studying higher order Lagrangian mechanics on Lie algebroids should be considered an interesting problem irrespective of any reduction. Indeed, via Theorem \ref{thrm:almost Lie} the results presented in the section will generalise quite directly to non-integrable Lie algebroids and almost Lie algebroids. It appears that we cannot directly generalise these constructions to skew algebroids, that is we cannot lose the compatibility of the anchor with the brackets. This compatibility of the anchor and the brackets also features as an essential ingredient in the variational approach developed by J\'{o}\'{z}wikowski \& Rotkiewicz \cite{Jozwikowski:2014}. We will present  in some detail the constructions for higher order Lagrangian mechanics on Lie algebroids as we expect this to be a particularly rich source of concrete examples and applications of our formalism.  Moreover, this situation leads to a `good' example and so we can derive Euler--Lagrange equations explicitly.

\subsection{Higher order Lagrangian mechanics on a Lie algebroid}
Let us consider a Lie groupoid $\mathcal{G}$ and a Lagrangian systems on $\textnormal{A}^{k}(\mathcal{G}) =  \left.\sT^{k}\mathcal{G}^{\:\underline{s}}\right|_{M}$.   We will refer to such systems as a \emph{k-th order Lagrangian system on the Lie algebroid} $\textnormal{A}(\mathcal{G})$ as the structure is completely defined by the underlying genuine  Lie algebroid structure on  $\textnormal{A}(\mathcal{G})$, see example \ref{eg:Lie algebroid epsilon}.  The relevant diagram here is

\begin{diagram}[htriangleheight=30pt ]
 \sT D^{*}(\textnormal{A}^{k}(\mathcal{G}))& \lTo^{\epsilon}& \sT^{*} D(\textnormal{A}^{k}(\mathcal{G}))  & \lRelation^{r} & \sT^{*}\textnormal{A}^{k}(\mathcal{G}) &   \\
             \dTo              &           &              \dTo                   &                 & \dTo     \uDotsto_{dL}  &     \\
 \sT\textnormal{A}(\mathcal{G})           &  \lTo^{\rho}    &            D(\textnormal{A}^{k}(\mathcal{G})   )          &  \lInto^{\iota}   & \textnormal{A}^{k}(\mathcal{G}) \\
\end{diagram}

\noindent  Here,
\begin{equation*}
D(\textnormal{A}^{k}(\mathcal{G}) ) \simeq \{ (Y,Z) \in \textnormal{A}(\mathcal{G}) \times_{M} \sT\textnormal{A}^{k-1}(\mathcal{G})| \bar{\rho}(Y) = \sT \bar{\tau}(Z)  \},
\end{equation*}
where $\bar{\rho} : \textnormal{A}(\mathcal{G}) \rightarrow \sT M$ is the standard anchor of the Lie algebroid and $\bar{\tau}: \textnormal{A}^{k-1}(\mathcal{G}) \rightarrow M$ is the obvious projection.\smallskip

Following example \ref{eg:Lie algebroid epsilon}, let us employ local coordinates $\left(X_{u}^{\mu} , Y_{U}^{I} , P_{\nu}^{k-1 -u} , \Pi_{J}^{k-1-U} \right)$ on $\sT^{*}D(\textnormal{A}^{k}(\mathcal{G}))$ and similarly let us employ local coordinates $\left( X_{u}^{\mu} , \Pi_{I}^{k-1-U}, \delta X_{U}^{\nu},  \delta \Pi_{J}^{U+1}\right)$ on $\sT D^{*}(\textnormal{A}^{k}(\mathcal{G}))$. Using \ref{eqn:lagrangian relation}, the relation $\Lambda_{L}^{\epsilon}: \textnormal{A}^{k}(\mathcal{G})  \rRelation \sT D^{*}(\textnormal{A}^{k}(\mathcal{G}))$ is given by
\begin{eqnarray}
\nonumber X_{u}^{\mu} &=& \bar{X}_{u}^{\mu},\\
\nonumber k \Pi_{i}^{1} &=& \frac{\partial L}{\partial \bar{X}^{i}_{k}}(\bar{x}),\\
\nonumber \delta X_{U}^{\nu} &=& U \: \rho[0]_{\nu}^{J} (\bar{x}) \bar{X}^{J}_{U},\\
\nonumber \delta \Pi_{I}^{U+1} &=& \rho[0]_{I}^{\mu}(\bar{x}) \left(\frac{\partial L}{\partial \bar{X}^{\mu}_{k-U}} - (k-U)\Pi_{J}^{U+1}  \right) + \delta_{U}^{k} \: C[0]_{J I }^{K }(\bar{x}) \bar{X}_{1}^{J}\Pi_{K}^{U}.
\end{eqnarray}
Recall that the only non-zero components of $\rho[0]_{I}^{\mu}$ are $\rho^{A}_{b}$ and $\delta_{b}^{a}$, while $C[0]_{I J }^{J}$ has only $C_{ab}^{c}$ non-vanishing and that $(\rho^{A}_{b}, C_{ab}^{c})$ are the structure functions of the Lie algebroid $\textnormal{A}(\mathcal{G})$. \smallskip

A curve $\zg(t)=( \bar{X}^{\za}_W(t))$ in $\textnormal{A}^{k}(\mathcal{G})$ satisfies the Euler-Lagrange equation if and only if it is $\zL_L^\ze$-related with an admissible curve (\ref{prolongation}). That is,
\begin{eqnarray}
\nonumber \frac{d}{dt}X_{U-1}^{\mu} &=& U\: \rho[0]_{J}^{\mu} (\bar{x}) \bar{X}^{J}_{U},\\
\nonumber k\: \Pi_{i}^{1} & =& \frac{\partial L}{\partial \bar{X}^{i}_{k}}(\bar{X})\\
\nonumber \frac{d}{dt} \Pi_{J}^{U} &=& \rho[0]_{J}^{\mu}(\bar{x}) \left(\frac{\partial L}{\partial \bar{X}^{\mu}_{k-U}} - (k-U)\Pi_{J}^{U+1}  \right) + \delta_{U}^{k} \: C[0]_{IJ}^{K }(\bar{x}) \bar{X}_{1}^{I}\Pi_{K}^{U}.
\end{eqnarray}
Note that we have a `good' case here, meaning that the  components of the anchor  $\rho_{a}^{b}$ are invertible, in this case trivially. For explicitness, let us revert back to coordinates $(\bar{x}^{A}, \bar{y}_{w}^{a} , \bar{z}^{b}_{k})$ on $\textnormal{A}^{k}(\mathcal{G})$. Recursively we can write the momenta $\pi_{a}$ in terms of the coordinates on  $\textnormal{A}^{k}(\mathcal{G})$ as follows; \smallskip

 \begin{tabular}{l}
                                                               $\displaystyle \pi_{a}^{1} = \frac{1}{k} \frac{\partial L}{\partial \bar{z}^{a}_{k}}$,\\
                                                               $\displaystyle (k-1)\pi_{b}^{2} = \frac{\partial L}{\partial \bar{y}^{b}_{k-1}} - \frac{1}{k}\frac{d}{dt}\left(  \frac{\partial L}{\partial \bar{z}^{b}_{k}} \right)$,\\
                                                               $\displaystyle (k-2)\pi_{c}^{3} =  \frac{\partial L}{\partial \bar{y}^{c}_{k-2}} - \frac{1}{(k-1)}\frac{d}{dt}\left(  \frac{\partial L}{\partial \bar{y}^{c}_{k-1}} \right) +\frac{1}{k(k-1)}\frac{d^{2}}{dt^{2}}\left(  \frac{\partial L}{\partial \bar{z}^{c}_{k}} \right) $\\
                                                               $\displaystyle \vdots $\\
                                                               $ \displaystyle\pi_{d}^{k} = \frac{\partial L}{\partial \bar{y}_{1}^{d}} - \frac{1}{2!}\frac{d}{dt}\left(\frac{\partial L}{\partial \bar{y}_{2}^{d}}\right) + \frac{1}{3!}\frac{d^{2}}{dt^{2}}\left(\frac{\partial L}{\partial \bar{y}_{3}^{d}}\right)- \cdots $\\
                                                               $\displaystyle + (-1)^{k} \frac{1}{(k-1)!} \frac{d^{k-2}}{dt^{k-2}}\left(\frac{\partial L}{\partial \bar{y}_{k-1}^{d}}\right) - (-1)^{k} \frac{1}{k!}\frac{d^{k-1}}{dt^{k-1}}\left(\frac{\partial L}{\partial \bar{z}_{k}^{d}}\right)$,
                                                              \end{tabular}\medskip

 \noindent which we recognise as the \emph{Jacobi--Ostrogradski momenta}. \smallskip

 The remaining equation  $\displaystyle \frac{d}{dt} \pi^{k}_{a} = \rho_{a}^{A}(\bar{x}) \frac{\partial L}{\partial \bar{x}^{A}} + \bar{y}_{1}^{b}C_{ba}^{c}(\bar{x})\pi_{c}^{k}$ can then  be written as
     \begin{equation*}
      \rho_{a}^{A}(\bar{x}) \frac{\partial L}{\partial \bar{x}^{A}} - \left(\delta^{c}_{a} \frac{d}{dt} - \bar{y}_{1}^{b}C_{ba}^{c}(\bar{x}) \right)\left(  \frac{\partial L}{\partial \bar{y}_{1}^{c}} - \frac{1}{2!}\frac{d}{dt}\left(\frac{\partial L}{\partial \bar{y}_{2}^{c}}\right) + \frac{1}{3!}\frac{d^{2}}{dt^{2}}\left(\frac{\partial L}{\partial \bar{y}_{3}^{c}}\right) \cdots - (-1)^{k}\frac{1}{k!}\frac{d^{k-1}}{dt^{k-1}}\left(\frac{\partial L}{\partial \bar{z}_{k}^{c}}\right)\right)=0,
     \end{equation*}

    \noindent which we define to be the \emph{k-th order Euler--Lagrange equations} on $\textnormal{A}(\mathcal{G})$, taking into account our choice of homogeneous coordinates.

\medskip
\noindent The above Euler-Lagrange equations are in complete agrement  with J\'{o}\'{z}wikowski \&   Rotkiewicz \cite{Jozwikowski:2014}, Colombo \& de Diego \cite{Colombo:2013} for the second order case, as well as Gay-Balmaz, Holm,  Meier,  Ratiu \&  Vialard's derivation of the higher Euler--Poincar\'{e} equations on a Lie group \cite{Gay-Balmaz:2012}. We clearly recover the standard higher Euler--Lagrange equations on $\sT^{k}M$ as a particular example. Note that the geometric structure on $\textnormal{A}^{k}(\mathcal{G})$ is completely encoded in the Lie algebroid $\textnormal{A}(\mathcal{G})$ and so the nomenclature we have chosen is appropriate. If we restrict ourselves to Lagrangians that do not depend on higher order (generalised) velocities then we recover the standard Euler--Lagrange equations on a a Lie algebroid, as first derived by Weinstein \cite{Weinstein:1996}  and generalized to (skew) algebroids in \cite{Grabowska:2006,Grabowska:2008}.

\subsection{Second order Hamel and Lagrange--Poincar\'{e}  equations}
In this subsection we briefly examine second order Lagrangian mechanics on the Atiyah algebroid and an Atiyah algebroid in the presence of a non-trivial connection. These examples give rise to generalisations of the Hamel and Lagrange--Poincar\'{e} equations.

\begin{example}\label{exm:Hamel}
An important situation in physics is when a Lagrangian defined on a principal $G$-bundle is invariant under the action of $G$. In this case if $P \rightarrow M$ is the principal bundle in question, then the Lagrangian is a function on $\sT^{k}P$. The reduction of this system is a Lagrangian system on $\textnormal{A}^{k} := \sT^{k}P\slash G$, which is to be considered as the  weighted (or higher order) version of the \emph{Atiyah algebroid}. Let us be a little more specific and consider $k=2$. Via a local trivialisation we can identify $\textnormal{A}^{2} \approx \sT^{2}M\times \mathfrak{g}[1] \times \mathfrak{g}[2]$ locally and thus  employ homogeneous coordinates $\{\bar{x}^{A}, \bar{v}_{1}^{B}, \bar{y}^{a}_{1}, \bar{w}_{2}^{A}, \bar{z}^{a}_{2} \}$. Then, in hopefully clear notation, the phase dynamics of a second order Lagrangian on an Atiyah algebroid is specified by

\renewcommand{\arraystretch}{2}
\begin{tabular}{lll}
$\displaystyle \delta x_{1}^{A} = \bar{v}_{1}^{A}$,  &  $\displaystyle \delta v_{2}^{A} = 2 \bar{w}_{2}^{A}$, & $\displaystyle \delta y_{2}^{a} = 2 \bar{z}^{a}_{2}$,\\
$\displaystyle \delta \pi_{A}^{3} = \frac{\partial L}{\partial \bar{x}^{A}}$, & $\displaystyle \delta \pi_{a}^{3} = \bar{y}^{c}_{1}C_{ca}^{b}\pi_{b}^{2}$, & $\displaystyle \delta \pi^{2}_{A} = \frac{\partial L}{\partial \bar{v}^{A}_{1}} - \pi_{A}^{2}$, \\
$\displaystyle \delta \pi_{a}^{2} = \frac{\partial L }{\partial \bar{y}_{1}^{a}} - \pi_{a}^{2} $, & $\displaystyle \pi_{A}^{1} = \frac{1}{2}\frac{\partial L}{\partial \bar{w}^{A}_{2}}$, & $\displaystyle \pi_{a}^{1} = \frac{1}{2} \frac{\partial L}{\partial \bar{z}^{a}_{2}}$.
\end{tabular}

\medskip
\noindent Via inspect we see that the phase dynamics is essentially separated into a part to do with the base $M$ and a part to do with the Lie group $G$. Thus, using the general result on mechanics on a Lie algebroid and  higher order tangent bundles as presented above, we arrive at the \emph{second order Hamel equations} \cite{Jozwikowski:2014} (also see \cite{Cendra:1998} for the first order case);

\begin{eqnarray}
\frac{\partial L}{\partial \bar{x}^{A}} - \frac{d}{dt}\left(\frac{\partial L}{\partial \bar{v}_{1}^{A}}  \right) + \frac{1}{2!} \frac{d^{2}}{dt^{2}}\left(\frac{\partial L}{\partial \bar{w}_{2}^{A}}  \right)  &=&0,\\
\nonumber  \left( \delta^{c}_{a} \frac{d}{dt} - \bar{y}^{b}_{1}C_{ba}^{c}\right)\left(\frac{\partial L}{\partial \bar{y}^{c}_{1}} - \frac{1}{2!} \frac{d}{dt}\left(\frac{\partial L}{\partial \bar{z}^{c}_{2}} \right) \right)&=&0.
\end{eqnarray}
\noindent  For the case of $\textnormal{A}^{2} = \sT^{2}G\slash G$ the above equations reduce to just the second equation which is the second order \emph{Euler--Poincar\'{e} equation}. The k-th order case follows directly.
\end{example}
\medskip

\begin{proposition}
Let $P \rightarrow M$ be a principal $G$-bundle, such that the Lie algebra $\mathfrak{g}$ of $G$ is abelian. Then given a higher order Lagrangian on the Atyiah algebroid $\textnormal{A}:= \sT P \slash G$ the momentum
\begin{equation*}
\pi_{a}^{k} = \frac{\partial L}{\partial \bar{y}_{1}^{d}} - \frac{1}{2!}\frac{d}{dt}\left(\frac{\partial L}{\partial \bar{y}_{2}^{d}}\right) + \cdots + (-1)^{k} \frac{1}{(k-1)!} \frac{d^{k-2}}{dt^{k-2}}\left(\frac{\partial L}{\partial \bar{y}_{k-1}^{d}}\right) - (-1)^{k} \frac{1}{k!}\frac{d^{k-1}}{dt^{k-1}}\left(\frac{\partial L}{\partial \bar{z}_{k}^{d}}\right),
\end{equation*}
\noindent is a constant of motion.
\end{proposition}

\begin{proof}
Follows directly from the k-th order Hamel equations upon setting $C_{ab}^{c} =0$.
\end{proof}

\begin{example}
For instance, let $L$ be the Lagrangian governing the motion of the tip of a javelin 
defined on $\sT^2\R^3$  by
$$
L(x,y,z)=\frac12\left(\sum_{i=1}^3(y^i)^2-(z^i)^2\right) \,.
$$
We can understand $G=\R^3$ here as a commutative Lie group, and since $L$ is $G$-invariant, we get immediately the reduction to the graded bundle $\R^3[1]\ti\R^3[2]$. The Euler-Lagrange equations
$$\frac{\xd}{\xd t}\left(\frac{\partial L}{\partial y^i} - \frac{1}{2}\frac{\xd}{\xd t}\left(\frac{\partial L}{\partial z^i}\right)\right)=0
$$
give in this case
$$\frac{\xd y^i}{\xd t}=\frac{1}{2}\frac{\xd^2 z^i}{\xd t^2}\,.
$$
The Lagrangian is regular and we get, similarly as in Example \,\ref{e-ham},
$\zl_L(y,z)=(y,-z)$, and
the Hamiltonian as a function on $\R^3[1]\ti\R^3[1]$ with coordinates $(y,\zvy)$:
$$H(y,\zvy)=\la\zvy,-\zvy\ran-L(y,-\zvy)=-\frac12\left(\sum_{i=1}^3(y^i)^2+(\zvy^i)^2\right)\,.
$$
Note that we do not have the minus sign if we wiew the Hamiltonian as a section of the corresponding AV-bundle.
\end{example}

\medskip
We have derived the Hamel equations using a geometric reduction and not a reduction of the variational problem.  In the context of the variational problem the  Hamel equations arise as a special case of the Lagrange--Poncar\'{e} equations in the presence of a trivial connection on $P \rightarrow P\slash G$ i.e. a trivial background Yang--Mills field.  To take a non-trivial background into account we have to \emph{deform}  the weighted Atyiah algebroid using the Yang--Mills field following \cite{Grabowska:2006}. Details for the second order case are presented in the proceeding example.

\begin{example}
Following example \ref{exm:Hamel}, but now in the presence of a non-trivial background Yang--Mills field, the weighted Atyiah algebroid of degree two becomes deformed by the connection and the associated curvature. This deformation only effects the $\delta \pi$ coordinates and leads  to a modified phase dynamics given by \medskip

\renewcommand{\arraystretch}{2}
\begin{tabular}{lll}
$\displaystyle \delta x_{1}^{A} = \bar{v}_{1}^{A}$,  &  $\displaystyle \delta v_{2}^{A} = 2 \bar{w}_{2}^{A}$, & $\displaystyle\delta \pi^{2}_{A} = \frac{\partial L}{\partial \bar{v}^{A}_{1}} - \pi_{A}^{2}$,\\
$\displaystyle \delta \pi_{A}^{3} = \frac{\partial L}{\partial \bar{x}^{A}} + \left(\bar{v}_{1}^{B} \mathbb{F}_{BA}^{a} + \bar{y}_{1}^{b}C_{bc}^{a}\mathbb{A}_{A}^{c} \right)\left( \pi_{a}^{2} + \bar{y}_{1}^{e}C_{ea}^{f}\pi_{f}^{1} \right)$, & $\displaystyle \delta y_{2}^{a} = 2 \bar{z}^{a}_{2}$, &
\\ $\displaystyle \delta \pi_{a}^{3} =  \bar{y}^{c}_{1}C_{ca}^{b}\pi_{b}^{2} + \bar{v}_{1}^{A}C_{ab}^{c}\mathbb{A}_{A}^{b}\pi_{c}^{2}$, &$\displaystyle \pi_{a}^{1} = \frac{1}{2} \frac{\partial L}{\partial \bar{z}^{a}_{2}}$, &\\
$\displaystyle \delta \pi_{a}^{2} = \frac{\partial L }{\partial \bar{y}_{1}^{a}}+ \bar{v}_{1}^{A}C_{ab}^{c}\mathbb{A}_{A}^{b}\pi_{c}^{1} - \pi_{a}^{2} $,&$\displaystyle \pi_{A}^{1} = \frac{1}{2}\frac{\partial L}{\partial \bar{w}^{A}_{2}}$ .&
\end{tabular}

\medskip
\noindent In the above  $\displaystyle \mathbb{A}_{A}^{a}$ are the components of the connection and $\displaystyle \mathbb{F}_{AB}^{a} = \frac{\partial \mathbb{A}_{B}^{a}}{\partial \bar{x}^{A}} -\frac{\partial \mathbb{A}_{A}^{a}}{\partial \bar{x}^{B}} + \mathbb{A}_{A}^{b}\mathbb{A}_{B}^{c}C_{cb}^{a}$ are  the components of the associated curvature. After a direct and  straightforward  calculation, the associated Euler--Lagrange equations are
\begin{eqnarray}
\nonumber \frac{\partial L}{\partial \bar{x}^{A}} - \frac{d}{dt}\left(\frac{\partial L}{\partial \bar{v}_{1}^{A}} \right) + \frac{1}{2}\frac{d^{2}}{dt^{2}}\left( \frac{\partial L}{\partial \bar{w}_{2}^{A}}\right) &=&  \left(\bar{v}_{1}^{B} \mathbb{F}_{BA}^{a} + \bar{y}_{1}^{b}C_{bc}^{a}\mathbb{A}_{A}^{c} \right) \frac{\partial L}{\partial \bar{y}_{1}^{a}}\\
 \nonumber &{-}& \left(\bar{v}_{1}^{B} \mathbb{F}_{BA}^{a} + \bar{y}_{1}^{b}C_{bc}^{a}\mathbb{A}_{A}^{c} \right)\left(\delta_{a}^{d}\frac{D}{Dt} + \bar{y}^{e}_{1}C_{ea}^{d} \right)\frac{1}{2}\frac{\partial L}{\partial \bar{z}_{2}^{d}}\,,\\
 \nonumber\left( \delta_{a}^{c}\frac{D}{Dt} - \bar{y}^{b}_{1}C_{ba}^{c} \right)\left( \frac{\partial L}{\partial \bar{y}^{c}_{1}} - \frac{1}{2}\frac{D}{Dt}\left(\frac{\partial L}{\partial \bar{z}_{2}^{c}}\right)\right) &=& 0\,,
\end{eqnarray}
\noindent where we have defined the covariant derivative $\displaystyle \frac{D}{Dt} \psi_{a} := \frac{d}{dt}\psi_{a} -  \bar{v}^{A}_{1}C_{ab}^{c}\mathbb{A}_{A}^{b} \psi_{c}$ for the appropriate objects.
\end{example}
The above pair of equations are, up to conventions,  the  \emph{second order Lagrange--Poincar\'{e}  equations} as defined in \cite{Gay-Balmaz:2011} using variational methods. It is clear that if we insist that the Lagrangian is independent of the weight two coordinates then we recover the classical Lagrange--Poincar\'{e} equations.   If the connection is trivial then we are back to the second order Hamel equations of the previous example.\medskip

\noindent If the Lie algebra $\mathfrak{g}$ is abelian then the second order Lagrange--Poincar\'{e} equations nicely simplify to
\begin{eqnarray}
\nonumber \frac{\partial L}{\partial \bar{x}^{A}} - \frac{d}{dt}\left(\frac{\partial L}{\partial \bar{v}_{1}^{A}} \right) + \frac{1}{2}\frac{d^{2}}{dt^{2}}\left( \frac{\partial L}{\partial \bar{w}_{2}^{A}}\right) &=&  \bar{v}_{1}^{B} \mathbb{F}_{BA}^{a}\left( \frac{\partial L}{\partial \bar{y}_{1}^{a}}
- \frac{d}{dt}\left(\frac{\partial L}{\partial \bar{z}_{2}^{a}}\right) \right),\\
\nonumber\frac{d}{dt}\left( \frac{\partial L}{\partial \bar{y}^{b}_{1}} - \frac{1}{2}\frac{d}{dt}\left(\frac{\partial L}{\partial \bar{z}_{2}^{b}}\right)\right) &=& 0.
\end{eqnarray}
\noindent The above equations describe a generalisation of the  equations describing the (non-relativistic) Lorentz force. Further notice that we have, in accordance to earlier observations that $\pi_{b}^{2} =\frac{\partial L}{\partial \bar{y}^{b}_{1}} - \frac{1}{2}\frac{d}{dt}\left(\frac{\partial L}{\partial \bar{z}_{2}^{b}}\right)$ is a constant of motion.


\vskip.3cm
\noindent Andrew James Bruce\\
\emph{Institute of Mathematics, Polish Academy of Sciences,}\\ {\small \'Sniadeckich 8,  00-656 Warszawa, Poland}\\ {\tt andrewjamesbruce@googlemail.com}\\

\noindent Katarzyna Grabowska\\
\emph{Faculty of Physics,
                University of Warsaw} \\
               {\small Pasteura 5, 02-093 Warszawa, Poland} \\
                 {\tt konieczn@fuw.edu.pl} \\

 \noindent Janusz Grabowski\\\emph{Institute of Mathematics, Polish Academy of Sciences}\\{\small \'Sniadeckich 8, 00-656 Warszawa,
Poland}\\{\tt jagrab@impan.pl}

\end{document}